\documentclass[12pt]{article}
\usepackage{amsfonts}
\usepackage{amsmath}
\usepackage{amsthm}
\usepackage{amssymb}
\usepackage{float}
\usepackage{graphicx}
\usepackage{times}
\usepackage{rotating,natbib}

\renewcommand{\baselinestretch}{1.5}

\newtheorem{theorem}{Theorem}
\newtheorem*{definition}{Definition}
\newtheorem{lemma}{Lemma}
\newtheorem*{remark}{Remark}
\newtheorem{proposition}{Proposition}

\newcommand{\eqschw}{\stackrel{d}{=}}

\begin{document}
\title{Asymptotic theory of range-based multipower variation}
\author{Kim Christensen\thanks{Corresponding author: Aarhus University, CREATES, Department of Economics and Business, University Park, Building 1322, 8000 Aarhus C, Denmark. Phone: (+45) 87 16 55 02, e-mail: \texttt{kchristensen@creates.au.dk}.} \and Mark Podolskij\thanks{University of Heidelberg, Institute of Applied Mathematics, Im Neuenheimer Feld 294, Room 215, 69120 Heidelberg, Germany. E-mail: \texttt{m.podolskij@uni-heidelberg.de}.}}
\date{This version: October, 2011}
\maketitle

\vspace*{-0.50cm}

\begin{abstract}
In this paper, we present a realized range-based multipower variation theory, which can be used to estimate return variation and draw jump-robust inference about the diffusive volatility component, when a high-frequency record of asset prices is available. The standard range-statistic -- routinely used in financial economics to estimate the variance of securities prices -- is shown to be biased when the price process contains jumps. We outline how the new theory can be applied to remove this bias by constructing a hybrid range-based estimator. Our asymptotic theory also reveals that when high-frequency data are sparsely sampled, as is often done in practice due to the presence of microstructure noise, the range-based multipower variations can produce significant efficiency gains over comparable subsampled return-based estimators. The analysis is supported by a simulation study and we illustrate the practical use of our framework on some recent TAQ equity data.


\end{abstract}

\thispagestyle{empty}
\newpage

\section{Introduction}
\setcounter{page}{1}
The standard, arbitrage-free continuous time setting for securities prices in financial economics shows that, in frictionless markets, return variation admits a general decomposition into a continuous, diffusive volatility component and discontinuous jumps \citep*[e.g.][]{andersen-bollerslev-diebold-labys:03a, back:91a, delbaen-schachermayer:94a}. In the past few years, our ability to assess the relative significance of these two fundamentally distinct sources of risk has taken a major step forward with the increasing availability and use of high-frequency data. This has opened the way for non-parametric estimation of the quadratic return variation and, for example, allows us to split this composite measure of risk into the integrated variance and the sum of the squared jumps. It follows a long-standing tradition within asset- and derivatives pricing, portfolio allocation and risk management of using low-frequency data (e.g., daily or weekly) to fit parametric stochastic volatility or jump-diffusion models \citep*[e.g.][]{alizadeh-brandt-diebold:02a, andersen-benzoni-lund:02a, bates:96a, bates:12a, chernov-gallant-ghysels-tauchen:03a, heston:93a, hull-white:87a, gallant-hsu-tauchen:99a}.

A large body of work in the high-frequency space is based on a framework called realized (return-based) multipower variation, which relies on statistics constructed from intraday returns \citep*[e.g.][]{ait-sahalia-jacod:09b, ait-sahalia-jacod:09a, ait-sahalia-jacod:11a, andersen-bollerslev-diebold:07a, barndorff-nielsen-shephard:04b, barndorff-nielsen-shephard:06a, barndorff-nielsen-shephard-winkel:06a, corsi-pirino-reno:10a, dobrev-szerszen:10a, huang-tauchen:05a, mancini:04a, mancini:09a, todorov:09a}.  In practice, the presence of market imperfections (such as price discreteness and bid-ask spreads) means that standard return-based multipower variation estimators often use sparse sampling \citep*[e.g.][]{bollerslev-todorov:11a, corsi-reno:12a, tauchen-zhou:11a}. The idea being that at a moderate frequency, for example the 5-minute frequency, the impact of microstructure noise is small enough to be ignored. Of course, for many liquid assets the data are much more abundant, so this principle often entails a significant loss of information, and much recent research has focused on developing estimators that are more resistant to the noise \citep*[e.g.][among others]{barndorff-nielsen-hansen-lunde-shephard:08a, christensen-oomen-podolskij:10a, fan-wang:07a, jacod-li-mykland-podolskij-vetter:09a, large:10a, podolskij-vetter:09a, zhang-mykland-ait-sahalia:05a}. A branch of related work can be found in, e.g., \citet*{andersen-dobrev-schaumburg:08a, lee-mykland:08a, lee-mykland:12a, mykland-shephard-sheppard:12a}.

In this paper, we formulate a complete realized \emph{range-based} multipower variation theory, which builds directly on the return-based multipower variation by simply replacing absolute returns with ranges, suitably scaled. We outline how it allows us to estimate aggregate return variation and form jump-robust estimates of integrated variance using the range. The analysis shows that if high-frequency data are being sparsely sampled, using a realized range-statistic can produce considerable efficiency gains relative to a standard return-based estimator, even when the latter employs subsampling to exhaust the entire database \citep*[e.g.][]{zhang-mykland-ait-sahalia:05a, zhou:96a}. Intuitively, the range partially distills some of the information contained in intermediate data not used by a sparsely sampled return-based estimator, and this turns out to be a more effective way of doing it compared to subsampling of low-frequency returns. Another appealing key feature of this theory is that realized range-based multipower variation estimators can be made increasingly robust to jumps without losing asymptotic efficiency. We use both simulations and empirical data to illustrate how these findings manifest, for example in order to construct feasible jump-robust confidence intervals for the integrated variance. Moreover, we show that this advantage largely prevails also in the presence of a realistic level of market microstructure noise. At low-frequency sampling, the range therefore offers a parsimonious, yet highly efficient, framework, which avoids the need for doing complicated corrections in order to combat the noise.

Indeed, the main motivation for using the range is that realized range-based estimation of the integrated variance is known to be very efficient in pure diffusion models \citep*[e.g.][]{christensen-podolskij:07a, martens-dijk:07a, parkinson:80a}. Interestingly, however, the properties of the high-low remain unchartered territory in jump-diffusion models and, as we show here, the standard range-estimator suffers from systematic biases, when jumps are added to the price equation. We propose to rectify this bias using a hybrid range-estimator, which has the form of a linear combination of the original range-statistic and a jump-robust measure of integrated variance.

It should be noted that the way we retrieve jump-robust measures of integrated multipower variation follows exactly the procedure of the return-based framework, which depends on terms in the proximity of a jump to get small. As such, it inherits some of the weaknesses associated with this approach, although our simulations do demonstrate that the range has superior finite sample robustness. In this respect, an interesting route is adopted by \citet*{dobrev:07a}, who extends the standard range-statistic (used here) -- based on the single largest price move -- into a generalized range theory, which maximizes the sum of multiple price moves. The generalized range is also jump-robust (potentially to some forms of infinite activity jump processes), but this feature descends from scaling constants, which may entail some advantages in finite samples. This topic is related to a multitude of recent alternative jump-robust estimators based on truncation \citep*[e.g.][]{ait-sahalia-jacod:09b, ait-sahalia-jacod:09a, andersen-dobrev-schaumburg:08a, christensen-oomen-podolskij:10a, christensen-oomen-podolskij:14a, corsi-pirino-reno:10a, mancini:04a, mancini:09a}.

The paper proceeds as follows. In section 2, we set notation and invoke a standard arbitrage-free continuous time jump-diffusion semimartingale model. We also briefly review some aspects of the theory of return-based multipower variation, before we switch to studying realized range-based multipower estimation. The key theoretical results are presented in Theorem \ref{Thm:RMV} and a novel combination estimator is proposed in Eq. \eqref{RBV_RTV}. In section 3, we conduct a Monte Carlo study to investigate the finite sample properties of our new range-based multipower variations. We also inspect the asymptotic approximation of the jump-robust range-based tripower variance. In section 4, we progress with some empirical results using high-frequency data from the TAQ database. In section 5, we conclude and offer directions for future research. An appendix contains the derivations of our theoretical results.

\section{Theoretical framework}
In this section, we derive new non-parametric theory, based on the price range, for consistently estimating return variation, and we show how it can be applied to filter out the continuous variation part from the squared jumps.

\subsection{The model}
The theory is developed for a univariate log-price, say $p = \left( p_{t} \right)_{t \geq 0}$, which is defined on a filtered probability space $\bigl( \Omega, \mathcal{F}, \left( \mathcal{F}_{t} \right)_{t \geq 0}, \mathbb{P} \bigr)$. $p$ evolves in continuous time and is adapted to the filtration $\left( \mathcal{F}_{t} \right)_{t \geq 0}$, which holds all relevant information released with the passing of time.

As standard in asset pricing theory, we assume that $p$ is a member of the class of jump-diffusion semimartingales that satisfy the generic representation:\footnote{Asset prices must be semimartingales under rather weak conditions \citep*[e.g.][]{back:91a, delbaen-schachermayer:94a}.}
\begin{equation}
\label{BSMJ} p_{t} = p_{0} + \int_{0}^{t} \mu_{u} \text{d}u + \int_{0}^{t} \sigma_{u} \text{d}W_{u} + \sum_{i = 1}^{N_{t}} J_{i},
\end{equation}
where $\mu = \left( \mu_{t} \right)_{t \geq 0}$ is a locally bounded and predictable drift term, $\sigma = \left( \sigma_{t} \right)_{t \geq 0}$ is a c\`{a}dl\`{a}g volatility process, $W = \left( W_{t} \right)_{t \geq 0}$ a standard Brownian motion, $N = \left( N_{t} \right)_{t \geq 0}$ a finite activity simple counting process, and $J = \left\{ J_{i} \right\}_{i = 1, \ldots, N_{t}}$ is a sequence of non-zero random
variables.\footnote{A simple counting process, $N$, is of finite activity provided that $N_{t} < \infty$ for $t \geq 0$, almost surely. In this paper, we do not explore infinite activity jump processes, although these models have been studied in the context of realized multipower variation \citep*[e.g.][]{ait-sahalia-jacod:09b, ait-sahalia-jacod:09a, barndorff-nielsen-shephard-winkel:06a, todorov-tauchen:10a, woerner:05a, woerner:06a}.} Here, $N$ represents the total number of jumps in $p$ that has occurred up to and including time $t$, while $J$ are the corresponding jump sizes. Note that the drift term is of order $\text{d}t$ and is therefore negligible over short intervals of time, as typically considered in the high-frequency literature. As such, the model induces two main sources of risk, namely diffusive volatility and jumps.

The quadratic variation of the cumulative return process is then given by
\begin{equation}
[\,p\,]_t = \int_0^t \sigma_u^2 \text{d}u + \sum_{i = 1}^{N_t} J_i^2,
\end{equation}
i.e. the integrated diffusive variance coefficient and the sum of the squared jumps. The quadratic variation plays a key role in high-frequency volatility estimation due to the following definition from stochastic integration theory:
\begin{equation}
\label{QvMath} \left[ \, p \, \right]_{t} = \underset{n \to \infty}{ \text{p-} \negmedspace \lim} \sum_{i = 1}^{n} (p_{t_{i}} - p_{t_{i - 1}})^{2},
\end{equation}
for any sequence of partitions $0 = t_{0} < t_{1} < \ldots < t_{n} = t$ such that $\max_{1 \leq i \leq n} \left\{ t_{i} - t_{i - 1} \right\} \to 0$ as $n \to \infty$ \citep*[e.g.][]{protter:04a}. This result is important, because it shows that we can infer the latent quadratic return variation from high-frequency observations of $p$ and, in the limit, consistently estimate it as more and more data are filled in the interval $[0,t]$. The practical relevance of the quadratic variation is stressed in several papers \citep*[e.g.][]{andersen-bollerslev-diebold:10a, andersen-bollerslev-diebold-labys:03a, barndorff-nielsen-shephard:07a}. It is, for example, closely related to the conditional variance, which features prominently in many pillars of financial economics.

In the special case where there are no jumps in $p$ (i.e. $N \equiv 0$), it holds that
\begin{equation}
\label{BSM} p_{t} = p_{0} + \int_{0}^{t} \mu_{u} \text{d}u + \int_{0}^{t} \sigma_{u} \text{d}W_{u},
\end{equation}
and the jump-diffusion process in Eq. \eqref{BSMJ} narrows down to a stochastic volatility model with continuous sample paths, for which the quadratic variation equals the integrated variance, $[\,p\,]_t = \int_0^t \sigma_u^2 \text{d}u $. We should note that some of our results are derived under this assumption.

Below, in the CLTs only, we are also going to impose some regularity conditions on $\sigma$:\\[-0.45cm]

\noindent \textbf{Assumption (V)}: \textit{$\sigma$ does not vanish \textbf{\upshape{(V$_{ \mathbf{1}}$)}} and it satisfies the equation:
\begin{equation}
\tag{ \textbf{V}$_{ \mathbf{2}}$} \sigma_{t} = \sigma_{0} + \int_{0}^{t} \mu_{u}^{ \prime} \text{\upshape{d}}u + \int_{0}^{t} \sigma_{u}^{ \prime} \text{\upshape{d}}W_{u} + \int_{0}^{t} v_{u}^{ \prime} \text{\upshape{d}}B_{u}^{ \prime}, \quad \text{\upshape{for
}} t \geq 0,
\end{equation}
where $\mu^{ \prime} = \left( \mu_{t}^{ \prime} \right)_{t \geq 0}$, $\sigma^{ \prime} = \left( \sigma_{t}^{ \prime} \right)_{t \geq 0}$ and $v^{ \prime} = \left( v_{t}^{ \prime} \right)_{t \geq 0}$ are c\`{a}dl\`{a}g, with $\mu^{ \prime}$ also being locally bounded and predictable, and $B^{ \prime} = \left( B_{t}^{ \prime} \right)_{t \geq 0}$ is a Brownian motion independent of $W$.}\\[-0.50cm]

Assumption (\textbf{V}$_{ \mathbf{1}}$) is a weak regularity condition, which is fulfilled for many financial models. Assumption (\textbf{V}$_{ \mathbf{2}}$) amounts to saying that $\sigma$ is of continuous semimartingale form.\footnote{Note the appearance of $W$ in $\sigma$, which allows for leverage effects \citep*[e.g.][]{christie:82a}.} It appears restrictive, because it rules out jumps in $\sigma$, which is at odds with empirical evidence \citep*[e.g.][]{eraker-johannes-polson:03a, todorov:10a}. But, it should be pointed out that the assumption is not a necessary condition and, hence, it is not required for our results to go through. It can be dispensed with in favor of a more flexible specification, allowing $\sigma$ to jump, at the cost of substantial extra rigor in the proofs \citep*[along the lines of][]{barndorff-nielsen-graversen-jacod-podolskij-shephard:06a}. Here, we rule out such technical details to preserve a leaner and more clear-cut exposition.

\subsection{The data}
Throughout the remainder of the text, we will work on the unit interval, $t \in [0,1]$. We think of this as representing the part of the day, for which high-frequency data are at our disposal. In our empirical application, where we are going to apply the estimators introduced in this section on a day-by-day basis to high-frequency data from NYSE- and NASDAQ-listed stocks, it is thus natural to let the unit interval be the hours spanned by the regular trading session.

The foundation for our econometric analysis is a high-frequency record of $p$, supposed to be available at equidistant times $t_{i} = i / N$, $i = 0, 1, \ldots, N$.\footnote{In practice, high-frequency data are irregularly spaced and equidistant prices are imputed from the observed ones. Two approaches are linear interpolation \citep*[e.g.][]{andersen-bollerslev:97b} or the previous-tick method suggested by \citet*{wasserfallen-zimmermann:85a}. The former method has an unfortunate property in connection with estimating quadratic return variation, see \citet*[][Lemma 1]{hansen-lunde:06a}.} Throughout, we assume $N = nm$, for $n, m \in \mathbb{N}$. Here, $n$ is the number of subintervals of the form $[(i - 1)/n,i/n]$, for $i = 1, \ldots, n$, on which we shall compute various statistics, while $m$ is the number of price changes within such a time interval.\footnote{This way of blocking high-frequency data into smaller pieces is natural in our setup, see, for example, \citet*{mykland:10a}.} Note that in the asymptotics we let $n \to \infty$, while, for simplicity, we keep $m$ fixed throughout.

Thus, the challenge is to infer the quadratic variation, and its split into the continuous and discontinuous part, from a set of discrete high-frequency data.

We define intraday returns at sampling frequency $n$ as follows
\begin{equation}
r_{i \Delta, \Delta} = p_{i / n} - p_{(i - 1) / n}, \quad \text{for } i = 1,\ldots, n,
\end{equation}
where $\Delta = 1 / n$ is the time distance between price observations.

Note that for $n < N$, the sequence $r_{i \Delta, \Delta}$ is not exhausting all the available high-frequency data. In effect, we are assuming that returns are constructed at a coarser sampling frequency $n$, relative to the total amount of $N$ ``ultra'' high-frequency returns that can potentially be made. In practice, high-frequency data are polluted with measurement error due to market microstructure frictions, such as price discreteness, bid-ask spreads etc. \citep*[e.g.][]{hansen-lunde:06a}. The presence of such noise can be detrimental to standard estimators of return variation, in particular at very high sampling frequencies. Although our ability to cope with and account for the impact of noise in the estimation of return variation has significantly improved in recent years\footnote{A representative, but necessarily incomplete, list of papers in this field, include \citet*[][]{barndorff-nielsen-hansen-lunde-shephard:08a, christensen-oomen-podolskij:10a, jacod-li-mykland-podolskij-vetter:09a, podolskij-vetter:09a, podolskij-vetter:09b, zhang-mykland-ait-sahalia:05a, zhang:06a}.}, it still remains common in applied work to use low-frequency tick data, for example 5-minute data are often used \citep*[e.g.][]{bollerslev-todorov:11a, corsi-reno:12a, tauchen-zhou:11a}. It is precisely this type of sparse sampling, which motivates us to suggest using the range-statistic below.

In this paper, we do not explicitly control for the noise.\footnote{\citet*{christensen-podolskij-vetter:09a} analyze the impact of noise on the range-statistic and propose a bias-correction to it.} Instead, we are going to assume that the econometrician has selected $n$ low enough such that potential biases from the noise can be ignored.

\subsection{Return-based estimation of quadratic variation}
With the setup in place, we can proceed by estimating quadratic variation and its two components using standard model-free return-based measures. The realized variance, proposed in \citet*{andersen-bollerslev:98a} and \citet*{barndorff-nielsen-shephard:02a}, is defined as a sum of squared intraday returns and is at sampling frequency $n$ given by:
\begin{equation}
\label{RVplim}
RV^{n} = \sum_{i = 1}^{n} r_{i \Delta, \Delta}^{2} \overset{p}{ \to} \int_{0}^{1} \sigma_{u}^{2} \text{d}u + \sum_{i = 1}^{N_{1}} J_{i}^{2}.
\end{equation}
The consistency of $RV^{n}$ for the quadratic return variation is immediate in light of Eq. \eqref{QvMath}. However, use of the realized variance by itself is not sufficient to learn about the composition of quadratic variation, so in order to isolate the integrated variance and squared jumps, we are going to require more material.

To accomplish this, consider the bipower and tripower variance:
\begin{align}
\label{BV_TV} BV^n &= \frac{n}{n - 1} \sum_{i = 1}^{n - 1} \prod_{j = 1}^{2} \frac{|r_{ \left(i + j - 1 \right) \Delta, \Delta}|}{\mu_{1}}, \\[0.25cm]
TV^n &= \frac{n}{n - 2} \sum_{i = 1}^{n - 2} \prod_{j = 1}^{3} \frac{|r_{ \left(i + j - 1 \right) \Delta, \Delta}|^{2/3}}{\mu_{2/3}},
\end{align}
where $\mu_r = E(|Z|^r)$ and $Z \sim N(0,1)$. The theoretical underpinnings of these estimators were derived in previous work by \citet*{barndorff-nielsen-shephard:04b} and \citet*{barndorff-nielsen-shephard-winkel:06a}, see also \citet*{barndorff-nielsen-graversen-jacod-podolskij-shephard:06a}.

In particular, they are consistent estimators of the integrated variance under both the jump-diffusion and pure diffusion model given by Eq. \eqref{BSMJ} and \eqref{BSM}, i.e. as $n \to \infty$
\begin{equation}
BV^n \overset{p}{\to} \int_{0}^{1} \sigma_{u}^{2} \text{d}u \qquad \text{and} \qquad TV^n \overset{p}{\to} \int_{0}^{1} \sigma_{u}^{2} \text{d}u.
\end{equation}
The intuition for the jump-robustness is that, with only a finite number of jumps in the log-price, all ``jump'' returns are eventually (i.e., for $n$ large enough) paired with a ``continuous'' return. The latter has order $O_p (\sqrt{\Delta})$, and this way jumps get knocked out of the probability limit.\footnote{There are a number of alternative ways of estimating the integrated variance robustly in the presence of jumps. \citet*{ait-sahalia-jacod:09b, ait-sahalia-jacod:09a} and \citet*{mancini:04a, mancini:09a}, for example, use threshold elimination of ``large'' returns before computing the realized variance, while \citet*{andersen-dobrev-schaumburg:08a} and \citet*{christensen-oomen-podolskij:10a} propose to infer diffusive volatility from the quantiles of
high-frequency returns. A couple of recent papers also show how to improve finite sample jump robustness of the bipower variance and how to make it more
efficient, see \citet*{corsi-pirino-reno:10a} and \citet*{mykland-shephard-sheppard:12a}.}

We can subsequently retrieve the sum of the squared jumps, e.g. by using $TV^n$
\begin{equation}
RV^n - TV^n \overset{p}{\to} \sum_{i = 1}^{N_{1}} J_{i}^{2}.
\end{equation}
Note that, in absence of price jumps, all the above three estimators are targeting the integrated variance. A natural way of comparing them in this special case is then done via their limiting distributions.

As an example, if $p$ follows the diffusion process defined by Eq. \eqref{BSM}, the CLT of $RV^n$ has the form
\begin{equation}
\label{Eqn:RVnCLT}
\sqrt{n} \biggl( RV^n - \int_{0}^{1} \sigma_{u}^{2} \text{d}u \biggr) \overset{d_{s}}{\to} MN \biggl(0, 2 \int_{0}^{1} \sigma_{u}^{4} \text{d}u \biggr),
\end{equation}
where $MN( 0, V)$ stands for a centered mixed normal distribution with conditional variance $V$, and $\int_{0}^{1} \sigma_{u}^{4} \text{d}u$ is the integrated quarticity.\footnote{Throughout the paper, the symbol ``$\overset{d_{s}}{\to}$'' is used to denote convergence in law stably. We refer to \citet*{renyi:63a} for a formal definition of stable convergence in law. Moreover, the motivation for using this type of convergence in the high-frequency volatility setting is explained in great detail by \citet*{barndorff-nielsen-hansen-lunde-shephard:08a}.}
The only way this result is altered for the $BV^n$ and $TV^n$ is that the factor multiplying the integrated quarticity increases to 2.6 and 3.1.

The key insight to draw from this discussion is that, in general, $RV^n$ estimates the quadratic return variation and, in the absence of jumps, it is most efficient.\footnote{In theory, the realized variance at the highest sampling frequency $N$ has the efficiency of the ML estimator in parametric versions of this problem. This shows that, in principle, we should construct the realized variance based on all data, but of course this result ignores problems associated with microstructure noise.} But $RV^n$ is not jump-robust, whereas both $BV^n$ and $TV^n$ are robust in the probability limit. Moreover, only the CLT of the tripower variance also holds true in the presence of jumps \citep*[e.g.][]{barndorff-nielsen-shephard-winkel:06a}.\footnote{This means that we can use the asymptotic distribution of $TV^n$ to construct confidence intervals for the integrated variance in both the jump and no-jump scenario, something we consider below.} Hence, using products of lagged returns (while keeping the sum of their powers equal to two) produces jump-robust estimates of integrated variance. The degree of robustness increases with the addition of extra lags, but we pay a price for this by losing some efficiency, if in fact there are no jumps. It is interesting to note this trade-off for later comparison, because the properties of the range-statistic turn out to be different.

\subsubsection{Subsampling}
In the above, the sparsely sampled return-based estimators potentially ignore a lot of information about the true return variation, because a large portion of the total amount of available high-frequency data is effectively discarded upfront. But, even if we do not want push the sampling frequency beyond $n$, it is easy to soak up more efficiency by subsampling the data, as suggested by \citet*[][]{zhou:96a} and \citet*[][]{zhang-mykland-ait-sahalia:05a}. This can be accomplished by simply shifting the starting point from which low-frequency returns are computed.

Let
\begin{equation}
r_{i \Delta, \Delta, j} = p_{i / n + j / N} - p_{(i - 1) / n + j / N}, \quad \text{for } i = 1,\ldots, n \text{ and } j = 0,1\ldots,m - 1.
\end{equation}
Then, we can compute and average $m$ realized variance estimates:
\begin{equation}
\label{RVsub}
SRV^{n,m} = \frac{1}{m} \sum_{j = 0}^{m - 1} RV^{n,j}, \quad \text{where} \quad RV^{n,j} = \sum_{i = 1}^{n} r_{i \Delta, \Delta, j}^2.\footnote{Notation is a bit loose here. We can only compute the $n$th return for the initial subsample estimate, while doing so for the rest requires a log-price observation outside of the unit interval. Thus, the last $m - 1$ subsample estimates are based on only $n - 1$ high-frequency returns and we apply an additional small sample correction to compensate for the missing summand.}
\end{equation}
The bipower and tripower variance can also be modified in this way and the subsampled version of these estimators will be called $SBV^{n,m}$ and $STV^{n,m}$ in the following.

The averaging performed by Eq. \eqref{RVsub} results in further efficiency gains, and it is known that the asymptotic variance factor of the subsampled realized variance can be brought down from 2 to 1.33 as the number of subsamples $m \to \infty$ \citep*[e.g.][]{zhang-mykland-ait-sahalia:05a}. To our knowledge, the reduction in variance associated with the subsampled bipower and tripower variance is not known in closed-form, but it will be accessed with Monte Carlo simulations below, where we compare it with the realized range-based estimators introduced next.

\subsection{Range-based estimation of quadratic variation}
The range provides an alternative way of learning about the quadratic return variation \citep*[e.g.][]{parkinson:80a}. Its use in the high-frequency context was initiated  by \citet*{christensen-podolskij:07a} and \citet*{martens-dijk:07a}, who developed the so-called realized range-based variance.

Write
\begin{equation}
\label{supremum} s_{p_{i \Delta, \Delta}, m} = \underset{0 \leq s, t \leq m}{ \max} \left( p_{\frac{i - 1}{n} + \frac{t}{N}} - p_{\frac{i - 1}{n} + \frac{s}{N}} \right), \qquad \text{for } i = 1,\ldots,n,
\end{equation}
as the range over the interval $[(i-1)/n,i/n]$.\footnote{Below, in the appendix of proofs, we also make use of the range of a standard Brownian motion over the interval $[(i-1)/n,i/n]$, which is denoted by $s_{W_{i \Delta, \Delta}, m}$, simply replacing $p$ with $W$ in the definition of Eq. \eqref{supremum}.} Then,
\begin{equation}
\label{RRVb}
RRV_{b}^{n, m} = \frac{1}{ \lambda_{2, m}} \sum_{i = 1}^{n} s_{p_{i \Delta, \Delta}, m}^{2},
\end{equation}
is the realized range-based variance at sampling frequency $n$. Here,
\begin{equation}
\label{sWm} \lambda_{r, m} = \mathbb{E} \bigl( s_{W, m}^{r} \bigr) \quad \text{with} \quad s_{W, m} = \underset{s, t = 0, 1, \ldots, m}{ \max \left\{ W_{t / m} - W_{s / m} \right\}}
\end{equation}
is the $r$th moment of the range of a standard Brownian motion on the unit interval $\left[ 0, 1 \right]$, where this expectation is based only on observations of the process at equidistant times $t_j = j/m$, for $j = 0,1,\ldots, m$.\footnote{A couple of points are worth highlighting here. First, the constants $\lambda_{r, m}$ used to rescale the range-statistic hinges on the assumption that data be equidistant, which is typically not the case in practice. We shall return to this below. Second, and in contrast to the return-based estimators, the range-based scalings are not available in closed-form. Although this is a drawback, it is relatively easy to estimate them by simulation. To facilitate this step, tables of $\lambda_{r, m}$, as a function of $r$ and $m$, can be obtained from the authors upon request.} The subscript $b$ appearing in the definition of $RRV_{b}^{n, m}$ indicates that it will be biased in the presence of jumps, as we detail below.

Assuming $p$ is a pure diffusion model as in Eq. \eqref{BSM}, and under condition (V), the main theoretical findings of \citet*{christensen-podolskij:07a} can be summarized as saying that
\begin{equation}
\label{RrvMn} \sqrt{n} \left( RRV_{b}^{n, m} - \int_{0}^{1} \sigma_{u}^{2} \text{\textup{d}}u \right) \overset{d_{s}}{ \to} MN \biggl( 0, \Lambda_{m} \int_{0}^{1} \sigma_{u}^{4} \text{\textup{d}}u \biggr),
\end{equation}
where $\Lambda_{m} = \left( \lambda_{4, m} - \lambda_{2, m}^{2} \right) / \lambda_{2, m}^{2}$.

The $\Lambda_{m}$ factor appearing in the CLT of $RRV_{b}^{n, m}$ depends on $m$, the total number of price changes available in each interval of the form $[(i-1)/n,i/n]$, $i = 1,\ldots,n$. We add that, for $m \geq 2$ as considered here, $\Lambda_{m}$ is always strictly smaller than the asymptotic variance coefficient of two for realized variance (see Figure \ref{Figure:avarRRV}).\footnote{As $m$ grows large, $\Lambda_{m}$ converges to a value of about 0.4, which is a restatement of the original result by \citet*{parkinson:80a}.} This comparison also extends to the subsampled version of realized variance. Hence, if data are recorded at a finer resolution than the sampling frequency, and so long as the influence of noise is not too severe, it is always better to construct a realized range-statistic than to compute realized variance. In effect, we are better able to recoup information about the integrated variance contained in intermediate data by computing a price range on that interval rather than subsampling low-frequency returns. But, it requires that $p$ follows the stochastic volatility model in Eq. \eqref{BSM} and it is therefore not valid in general. Still, it highlights the potential of the range and motivates us to analyze its properties in the jump-diffusion context.

\subsubsection{Extension to jump-diffusion processes}
To the best of our knowledge, little is known about the high-low estimator in models with jumps, as defined by Eq. \eqref{BSMJ}. It turns out that in its raw form, the realized range-based variance is inconsistent for the quadratic return variation, if there are jumps in the price process, as highlighted in Theorem \ref{IcRrv}.
\begin{theorem}
\label{IcRrv} Assume that $p$ follows the jump-diffusion process defined in Eq. \eqref{BSMJ}. As $n \to \infty$, it holds that:
\begin{equation}
RRV_{b}^{n, m} \overset{p}{ \to} \int_{0}^{1} \sigma_{u}^{2} \text{\textup{d}}u + \frac{1}{ \lambda_{2, m}} \sum_{i = 1}^{N_{1}} J_{i}^{2},
\end{equation}
\end{theorem}
\begin{proof}
See appendix.
\end{proof}
As the theorem shows, $RRV_{b}^{n, m}$ is downward biased in the jump-diffusion framework, because it incorrectly scales down the squared jumps by $\lambda_{2,m}$. To understand this, consider a model that consists solely of jumps (i.e., with drift and diffusion coefficient set to zero). Then, as $n \to \infty$, the sum of squared ranges converges to the sum of squared jumps, so the re-scaling is not required.

Although the conclusion of Theorem \ref{IcRrv} is disappointing, the structure of the inconsistency unveiled by it does suggest a quick fix for constructing a hybrid range-statistic that can estimate the whole quadratic return variation.

We could, for example, do the following:
\begin{equation}
\lambda_{2, m} RRV_{b}^{n, m} + \left(1 - \lambda_{2, m} \right) STV^{n,m} \overset{p}{ \to} \int_{0}^{1} \sigma_{u}^{2} \text{\textup{d}}u + \sum_{i = 1}^{N_{1}} J_{i}^{2},
\end{equation}
i.e., we take a linear combination of $RRV_{b}^{n, m}$ and $STV^{n,m}$ using the weights $(\lambda_{2, m}, 1 - \lambda_{2, m})$. This effectively amounts to undoing the re-scaling of the squared ranges in Eq. \eqref{RRVb} and then using a jump-robust estimator to subtract the excess portion of the integrated variance that emanates from this operation.

\subsubsection{realized range-based multipower variation}
The main problem with this approach is that in doing the bias-correction to the realized range-based variance, we are relying on a return-based estimator as the robust measure of integrated variance, which conflicts with our intention of using the range-statistic. As such, a jump-robust range-based estimator of the integrated variance is required, which, again to best of our knowledge, has not been proposed in the literature. To fill this hole, we are therefore going to introduce and study a complete range-based multipower variation theory, as laid out next.

\begin{definition} The realized range-based multipower variation with parameter $\left(q_{1}, \ldots, q_{k} \right) \in \mathbb{R}_{+}^{k}$ is defined as:
\begin{equation}
\label{Rmv} RMV_{\left( q_{1}, \ldots, q_{k} \right)}^{n,m} = \frac{n}{n - k + 1} n^{q_{+} / 2 - 1} \sum_{i = 1}^{n - k + 1} \prod_{j = 1}^{k} \frac{s_{p_{(i + j - 1) \Delta, \Delta}, m}^{q_{j}}}{\lambda_{q_{j},m}},
\end{equation}
where $q_{+} = \sum_{j = 1}^{k} q_{j}$.
\end{definition}

$RMV_{\left( q_{1}, \ldots, q_{k} \right)}^{n,m}$ is composed of suitably scaled range-based cross-products raised to the powers $\left(q_{1}, \ldots, q_{k} \right)$ and it constitutes a direct analogue to the general definition of realized multipower variation \citep*[e.g.][]{barndorff-nielsen-shephard-winkel:06a}.

\begin{remark}
Note that $s_{p_{i \Delta, \Delta}, m} \geq 0$ and so $RMV_{\left( q_{1}, \ldots, q_{k} \right)}^{n,m}$ is non-negative by construction.
\end{remark}

We should point out that, with proper modifications throughout, $(i + j - 1)$ may be replaced by $(i + Kj - 1)$ for any finite positive integer $K$ in the definition of $RMV_{\left( q_{1}, \ldots, q_{k} \right)}^{n,m}$. Such "staggering" of the data has been suggested in \citet*{andersen-bollerslev-diebold:07a} and \cite{barndorff-nielsen-shephard:06a}. Moreover, \citet*{huang-tauchen:05a} show how extra lagging can help to reduce the impact of microstructure noise in this type of estimators by effectively breaking the serial correlation in returns induced by the noise.

The next result states the theoretical properties of the realized range-based multipower variation.

\begin{theorem}
\label{Thm:RMV}
Assume that $p$ follows the diffusion process defined in Eq. \eqref{BSM}. As $n \to \infty$, it holds that:
\begin{equation}
\label{Eqn:PlimRMV}
RMV_{\left( q_{1}, \ldots, q_{k} \right)}^{n,m} \overset{p}{ \to} \int_{0}^{1} |\sigma_{u}|^{q_{+}} \text{\upshape{d}}u.
\end{equation}
Moreover, if condition (V) is satisfied, it additionally holds that
\begin{equation}
\label{Eqn:CltRMV}
\sqrt{n} \left( RMV_{\left( q_{1}, \ldots, q_{k} \right)}^{n,m} - \int_{0}^{1} |\sigma_{u}|^{q_{+}} \text{\upshape{d}}u \right) \overset{d_{s}}{ \to}
MN \biggl( 0, \Lambda_{\left(q_{1}, \ldots, q_{k} \right)}^{m} \int_{0}^{1} |\sigma_u|^{2q_{+}} \text{\upshape{d}}u \biggr),
\end{equation}
where
\begin{equation*}
\Lambda_{\left(q_{1}, \ldots, q_{k} \right)}^{m} = \frac{\displaystyle \prod_{j = 1}^{k} \lambda_{2q_{j}, m} - (2k - 1) \prod_{j = 1}^{k} \lambda_{q_{j}, m}^{2}
+ 2 \sum_{h = 1}^{k - 1} \prod_{j = 1}^{h} \lambda_{q_{j}, m} \prod_{j = k - h + 1}^{k} \lambda_{q_{j}, m} \prod_{j = 1}^{k - h} \lambda_{q_{j} + q_{j + h}, m}}{ \displaystyle \prod_{j = 1}^{k} \lambda_{q_{j},m}^{2}}.
\end{equation*}
Furthermore, the consistency result in Eq. \eqref{Eqn:PlimRMV} is robust to jumps if $\underset{1 \leq j \leq k}{ \max} (q_{j}) < 2$, while the CLT in Eq. \eqref{Eqn:CltRMV} is robust to jumps under the stronger condition $\underset{1 \leq j \leq k}{ \max} (q_{j}) < 1$.
\end{theorem}
\begin{proof}
See appendix.
\end{proof}

\begin{remark}
Note that the rate of convergence is not influenced by $m$ and no assumptions on the ratio $n / m$ are required.
\end{remark}

Theorem \ref{Thm:RMV} lays the foundation for producing range-based estimates of integrated power variation of various orders. It also shows what is required for such estimates to be robust against jumps in their probability limit and asymptotic distribution.

In light of the return-based multipower variation theory, there is nothing too surprising about the conclusions of the theorem. As such, the interested reader should note, while going through the appendix, that the recipe used to prove the results is to some extent ``standard'' by now. This means that several of the steps taken to deduce the properties of $RMV_{\left( q_{1}, \ldots, q_{k} \right)}^{n,m}$ borrow directly from extant literature \citep*[e.g.][]{barndorff-nielsen-graversen-jacod-podolskij-shephard:06a, christensen-podolskij:07a}. However, the range is a complicated functional, which has a number of subtle, technical implications in the analysis. In the proofs, we highlight where these complications arise and we also try to pinpoint what is ``new'' relative to the existing theory.

In this paper, the full force of Theorem \ref{Thm:RMV}, which can of course be used to estimate many interesting objects, is not required. We will mainly focus on defining jump-robust realized range-based estimates of the integrated variance by cloning the return-based bipower and tripower variance.\footnote{Below, a jump-robust estimator of the integrated quarticity is also used.} Thus, we define
\begin{align}
\label{RBV_RTV} &RBV^{n, m} = \frac{n}{n - 1} \sum_{i = 1}^{n - 1} \prod_{j = 1}^{2} \frac{s_{p_{(i + j - 1) \Delta, \Delta}, m}}{\lambda_{1, m}}, \text{ i.e. } q_{1} = q_{2} = 1,\\[0.25cm]
&RTV^{n,m} = \frac{n}{n - 2} \sum_{i = 1}^{n - 2} \prod_{j = 1}^{3} \frac{s_{p_{(i + j - 1) \Delta, \Delta}, m}^{2/3}}{\lambda_{2/3, m}}, \text{ i.e. } q_{1} = q_{2} = q_{3} = 2/3,
\end{align}
for which it holds that
\begin{equation}
RBV^{n, m} \overset{p}{ \to} \int_{0}^{1} \sigma_{u}^{2} \text{d}u \quad \text{and} \quad RTV^{n, m} \overset{p}{ \to} \int_{0}^{1} \sigma_{u}^{2} \text{d}u,
\end{equation}
while only the limiting distribution of the $RTV^{n, m}$ remains unchanged in the jump setting, because the maximum of its powers is strictly smaller than unity.

Figure \ref{Figure:avarRRV} plots the asymptotic variance coefficient $\Lambda_{\left(q_{1}, \ldots, q_{k} \right)}^{m}$, as a function of $m$, for the three range-based estimators considered here.\footnote{Note that the special case $m = 1$ means that we are using only $p_{(i-1) / n}$ and $p_{i/n}$ to compute the range-statistic on $[(i-1)/n,i/n]$. With no interior data available, the absolute return and the range are identical. This explains why the variance of the realized range-based estimators, for $m = 1$, coincides with those known from the return-based multipower variation theory, see, e.g., Eq. \eqref{Eqn:RVnCLT} and the following discussion.} As evident, the variance decreases monotonically with increasing $m$, and by the time $m$ reaches ten, a large majority of the potential efficiency gain has been attained. To put this in perspective, consider using the popular 5-minute sampling frequency. Then, $m = 10$ is equivalent to actually observing the price every 30 seconds, which is not unrealistic for many liquid series.

\begin{figure}[t!]
\begin{center}
\caption{Asymptotic variance factor of realized range-based estimators.
\label{Figure:avarRRV}}
\begin{tabular}{c}
\includegraphics[height=8.00cm,width=0.75\linewidth]{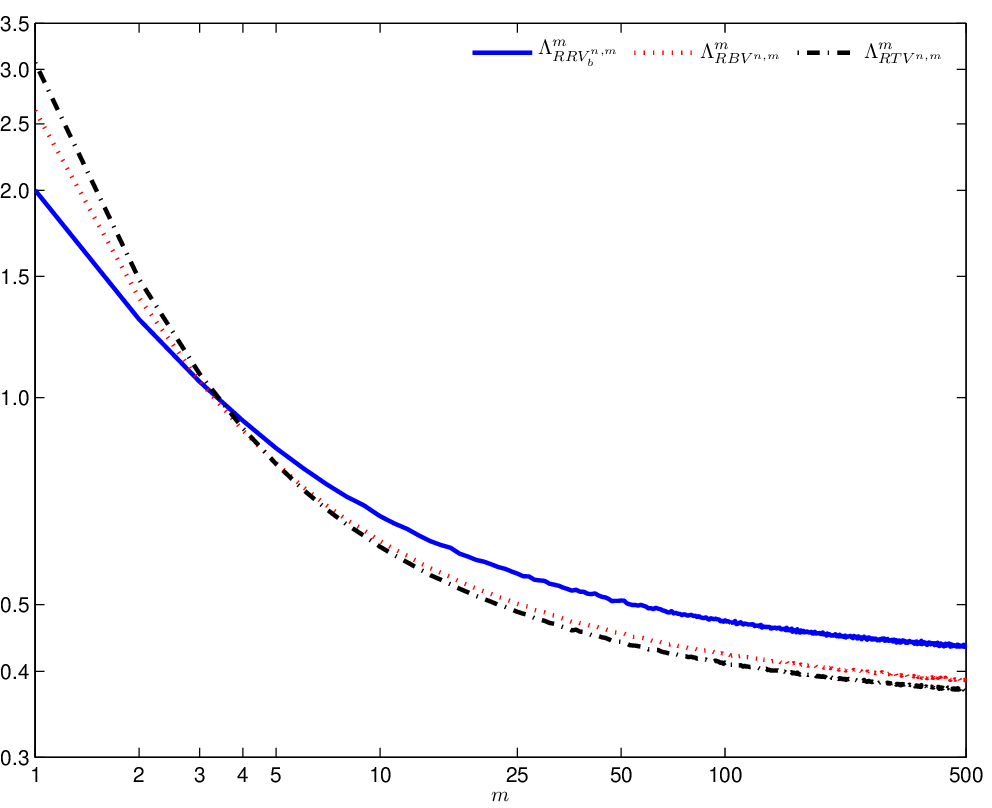}
\end{tabular}
\begin{scriptsize}
\parbox{0.925\textwidth}{\emph{Note.} The figure shows the asymptotic variance factor of the range-based realized, bipower and tripower variance, as a function of $m$.}
\end{scriptsize}
\end{center}
\end{figure}

How much efficiency is being sacrificed to obtain jump robustness in the range-based setting? The answer, which can be gauged from the figure, is quite surprising. Note that, for $m \leq 3$, the ordering of the estimators is as expected, with $RRV_{b}^{n,m}$ being the most efficient. However, and very intriguingly, the rankings of the variances are swapped for $m \geq 4$, rendering the realized range-based tripower variance not only the more robust estimator of the integrated variance but also the most efficient!\footnote{Taken together, the above theory implies that in many practical cases, the $RRV_{b}^{n, m}$ of \citet*{christensen-podolskij:07a} and \citet*{martens-dijk:07a} should not be applied as a standalone estimator. In the presence of jumps, it is a biased estimate of quadratic variation, and here we should subsume $RRV_{b}^{n, m}$ into the combined estimator given by Eq. \eqref{ConsQV} below, while, in the absence of jumps, the variance of $RRV_{b}^{n, m}$ is inferior to the range-based multipower variation alternatives, if $m \geq 4$.}

While the message conveyed by Figure \ref{Figure:avarRRV} is compelling, it also clashes with intuition, because the range-based multipower estimators are designed to be increasingly robust to jumps. We would expect this feature to come at a cost, if there are no jumps in the data, vis-\`{a}-vis the trade-off embedded in the return-based estimators. In general, the efficiency of a multipower variation statistic is a property of the underlying moments of Brownian motion of the given functional, i.e. in our setting the range. So, there is nothing to stop a multipower variation statistic with extra lags from being more efficient. But, apart from saying that in the proofs the constants, which pop up here and there, make it a fact, we lack a convincing, intuitive explanation about, why it is true. The marginal reduction in variance from adding extra lags disappears fast, though, and from a practical perspective the range-based tripower variance appears sufficient to capture almost all incremental efficiency gain.

Using these results, we will close this subsection by introducing a new, purely range-based estimator that is consistent for the quadratic variation of the jump-diffusion semimartingale defined by Eq. \eqref{BSMJ}:
\begin{equation}
\label{ConsQV} RRV^{n, m} \equiv \lambda_{2, m} RRV_{b}^{n, m} + \left( 1 - \lambda_{2, m} \right) RTV^{n, m} \overset{p}{
\to} \int_{0}^{1} \sigma_{u}^{2} \text{\textup{d}}u + \sum_{i = 1}^{N_{1}} J_{i}^{2},
\end{equation}
which is our preferred approach of estimating total return variation with the range.

\subsubsection{The joint distribution of $RRV_{b}^{n,m}$ and $RTV^{n,m}$ }
The univariate convergence in law of $RRV_{b}^{n,m}$ and $RTV^{n,m}$, which is available from Theorem \ref{Thm:RMV}, can be expanded to cover their joint asymptotic distribution, which is required as a basic ingredient for conducting range-based non-parametric tests for the presence of jumps, a topic that has attracted considerable attention in past research \citep*[e.g.][]{ait-sahalia-jacod:09b, ait-sahalia-jacod:09a, ait-sahalia-jacod:11a, barndorff-nielsen-shephard:04b, barndorff-nielsen-shephard:06a, christensen-oomen-podolskij:14a, huang-tauchen:05a, jiang-oomen:08a, lee-mykland:08a, li:11a}.

\begin{proposition}
\label{RrVRbVCiD} Assume that $p$ follows the diffusion process defined in Eq. \eqref{BSM}, where the conditions on $\sigma$ given by (V) are satisfied. As $n \to \infty$, it holds that:
\begin{equation}
\sqrt{n} \left( \begin{array}{c}
RRV_{b}^{n, m} - \displaystyle \int_{0}^{1} \sigma_{u}^{2}
\text{\textup{d}}u \\[0.25cm]
RTV^{n, m} - \displaystyle \int_{0}^{1} \sigma_{u}^{2}
\text{\textup{d}}u
\end{array}
\right) \overset{d_{s}}{ \to} \negmedspace MN \left( \mathbf{0},
\int_{0}^{1} \sigma_{u}^{4} \text{\textup{d}}u \left[
\begin{array}{ll}
\Lambda_{RRV_{b}^{n,m}}^{m} & \Lambda_{RRV_{b}^{n,m}, RTV^{n,m}}^{m} \\[0.25cm]
\Lambda_{RRV_{b}^{n,m}, RTV^{n,m}}^{m} & \Lambda_{RTV^{n,m}}^{m}
\end{array} \right] \right),
\end{equation}
with
\begin{equation}
\Lambda_{RRV_{b}^{n,m}, RTV^{n,m}}^{m} = \frac{2}{ \lambda_{1, m}^{2} \lambda_{2/3, m}^{3}} \left( \lambda_{1, m} \lambda_{5/3, m}^{2} \lambda_{2/3, m}^{2} + \lambda_{5/3, m}^{4} \lambda_{2/3, m} - 2 \lambda_{1, m}^{2} \lambda_{2/3, m}^{3} \right).
\end{equation}
\end{proposition}
Because the jump detection analysis is not implemented in this paper, we exclude a formal verification of the proposition (the proof can be forwarded upon request).\footnote{Previous drafts of this article dealt more formally with the range-based jump detection analysis. An electronic copy of this material can be found on the web or acquired by emailing the authors.} But, it should be noted that the bivariate extension is relatively simple to derive.

\section{Simulation study}
In this section, we document some aspects of the above asymptotic analysis by means of Monte Carlo experiments. The purpose of the study is to understand whether the theoretical, large sample properties of the realized range-based estimators are preserved in smaller, but more realistic, sample sizes.

For the continuous piece of the model, we adopt a dynamic two-factor stochastic volatility process, which can generate highly erratic sample paths for the log-price and volatility series. It is based on previous empirical work carried out by \citet*{chernov-gallant-ghysels-tauchen:03a}, so that our setup ensures that the simulation design captures many salient features of real equity data (e.g., leverage correlation and volatility feedback); a market considered in our empirical application below.

In particular, the first building block is as follows:
\begin{equation}
\label{Eqn:simP}
\text{d}p_{t} = \mu \text{d}t + \text{s-exp} \left[ \beta_0 + \beta_1 \sigma_{t}^{(1)} + \beta_2 \sigma_{t}^{(2)} \right] \text{d}W_{t},
\end{equation}
with
\begin{equation}
\label{Eqn:simV}
\begin{array}{r@{~}l}
\text{d} \sigma_{t}^{(1)} &= \alpha_1 \sigma_{t}^{(1)} \text{d}t + \text{d}B_{t}^{(1)}, \\[0.25cm]
\text{d} \sigma_{t}^{(2)} &= \alpha_2 \sigma_{t}^{(2)} \text{d}t + \left[ 1 + \alpha_3 \sigma_{t}^{(2)} \right] \text{d}B_{t}^{(2)}.
\end{array}
\end{equation}
Here, $\text{s-exp}(x)$ is the so-called ``spliced'' exponential function.\footnote{We refer to \citet*{chernov-gallant-ghysels-tauchen:03a} for more information about how the $\text{s-exp}(x)$ function operates. In short, $\text{s-exp}(x)$ slows down the growth rate of the exponential function at high values of the input $x$.} The parameter values for the entire system are taken from \citet*{huang-tauchen:05a}, i.e. $\left( \mu, \beta_{0}, \beta_{1}, \beta_{2}, \alpha_{1}, \alpha_{2}, \alpha_{3} \right) = \left( 0.03, -1.2, 0.04, 1.5, -0.000137, -1.386, 0.25 \right)$ and $\text{corr}(\text{d}W_{t}, \text{d}B_{t}^{(1)}) = \text{corr}(\text{d}W_{t}, \text{d}B_{t}^{(2)}) = -0.3$.

To specify the discontinuous piece of $p$, a compound Poisson process is used. We fix a constant intensity parameter $\kappa = 0.4$ per time unit. Thus, a jump in $p$ is experienced every 2.5 replication, on average, and we draw the corresponding jump sizes from a normal distribution, $J_i \sim N \left( 0, p_{jmp} \int_0^1 \sigma_{u}^{2} \text{d}u \right)$. In our simulations, we use $p_{jmp} = 0.25 / n_J$, where $n_J$ is the total number of jumps in a given simulation run. Quadratic jump variation is thus taken to be proportional to the integrated variance, with a typical squared jump being larger in size on high volatility days. Together with our selection of $\kappa$, this choice implies that the unconditional jump proportion is about 8\% of total return variation; a figure that broadly agrees with a consensus measure from the extant, recent literature \citep*[e.g.][]{andersen-bollerslev-huang:11a, bollerslev-law-tauchen:08a, corsi-reno:12a, tauchen-zhou:11a, todorov:09a}.

A standard Euler approximation scheme is applied to the set of stochastic differential equations given by Eq. \eqref{Eqn:simP} -- \eqref{Eqn:simV}. Jumps are scattered randomly throughout the day. We process a total of 10,000 simulations and assume that $N = 2,340$. The latter choice is calibrated to match our empirical work in Section \ref{Section:Empiri}, where we study data from the NYSE TAQ database. The regular trading session at NYSE spans 6.5 hours, or 23,400 seconds, and for the sample period covered, we refresh the price every 10 seconds, which motivates our selection of $N$ here.\footnote{In order to minimize discretization bias, we first create a compact realization of Eq. \eqref{Eqn:simP} -- \eqref{Eqn:simV} based on simulating a total of 23,400 ``second-by-second'' log-price updates and from this we extract every 10th data point.} We use the sampling frequencies $n = 26, 39, 78$, which translates into  15-, 10- and 5-minute sampling. Also, we construct return-based estimators by using subsampling, as explained above.

Finally, and although we do not model microstructure noise explicitly in this paper, our approach is motivated by the existence of such frictions. Therefore, we also gauge the performance of our estimators in the presence of noise. In particular, we add to $p$ an i.i.d. noise process $u$ -- independent of $p$ -- such that $\mathbb{E} \left( u \right) = 0$ and $\mathbb{E} \left( u^{2} \right) = \varpi^{2}$. We let $u$ have a two-point distribution: $\text{Pr} \left( u = \pm \varpi\right) = 1/2$. This choice has been further analyzed in \citet*{christensen-podolskij-vetter:09a} and can loosely be thought of as representing a form of bid-ask spread. The magnitude of the noise is controlled by $\displaystyle \gamma = \sqrt{\frac{\varpi^{2} \int_{0}^{1} \sigma_{u}^{2} \text{d}u}{N}}$, where $\gamma$ is the noise ratio parameter \citep*[see, e.g.,][]{oomen:06a}. \citet*{christensen-oomen-podolskij:10a} report a comprehensive set of empirical $\gamma$ estimates. In accordance with their results, we set $\gamma = 0.50$, which reflects the typical amount of noise found in high-frequency data from the U.S. stock market.

\subsection{Simulation results}

We start by looking closer at the ability of the range-based estimators to provide unbiased and efficient measures of quadratic variation and integrated variance. In Table \ref{Table:SimRMSE}, we report the relative bias and root mean squared error (rmse) of the various statistics for the three sampling frequencies $n = 26, 39, 78$. The relative bias is computed as a ratio of the estimate to its population target, averaged across simulations, and should equal 1 for an unbiased statistic. The number in parenthesis below the relative bias is the rmse, which has been multiplied by a factor 1,000.

\begin{table}[H]
\setlength{\tabcolsep}{0.35cm}
\begin{center}
\caption{Relative bias and root mean squared error.
\label{Table:SimRMSE}}
\smallskip
\begin{tabular}{lcccccccc}
\hline \hline
&& \multicolumn{3}{c}{\emph{Absence of noise}} && \multicolumn{3}{c}{\emph{Presence of noise}} \\
\cline{3-5} \cline{7-9}
& $(n,m) = $ & $(78,30)$ & $(39,60)$ & $(26,90)$ && $(78,30)$ & $(39,60)$ & $(26,90)$  \\
\hline
\multicolumn{6}{l}{\emph{Panel A: Range-based}} \\
$RRV^{n,m}$ & [QV]       &  1.011    &   1.022   &   1.033   &&  1.018   &   1.027   &   1.037  \\[-0.25cm]
            &            & (0.349)   &  (0.522)  &  (0.667) && (0.353)  &  (0.524)  &  (0.669) \\ \\
$RBV^{n,m}$ & [IV]       &  1.019    &   1.023   &   1.025  &&  1.031   &   1.031   &   1.032  \\[-0.25cm]
            &            & (0.298)   &  (0.398)  &  (0.485) && (0.310)  &  (0.405)  &  (0.492) \\ \\
$RTV^{n,m}$ & [IV]       &  1.010    &   1.011   &   1.010  &&  1.022   &   1.019   &   1.017  \\[-0.25cm]
            &            & (0.281)   &  (0.378)  &  (0.469) && (0.291)  &  (0.384)  &  (0.476) \\ \\
\multicolumn{6}{l}{\emph{Panel B: Return-based (with subsampling)}} \\
$SRV^{n,m}$ & [QV]       &  0.999    &   0.999   &   0.997  &&  1.001   &   0.999   &   0.997  \\[-0.25cm]
            &            & (0.467)   &  (0.663)  &  (0.816) && (0.467)  &  (0.662)  &  (0.817) \\ \\
$SBV^{n,m}$ & [IV]       &  1.034    &   1.039   &   1.040  &&  1.036   &   1.040   &   1.040  \\[-0.25cm]
            &            & (0.539)   &  (0.749)  &  (0.901) && (0.540)  &  (0.749)  &  (0.901) \\ \\
$STV^{n,m}$ & [IV]       &  1.021    &   1.023   &   1.022  &&  1.023   &   1.023   &   1.022  \\[-0.25cm]
            &            & (0.528)   &  (0.737)  &  (0.895) && (0.529)  &  (0.737)  &  (0.895) \\
\hline \hline
\end{tabular} \smallskip
\begin{footnotesize}
\parbox{0.925\textwidth}{\emph{Note}. We report the relative bias and rmse of the estimators included in the simulation study. The bias measure is equal to 1 for an unbiased estimator. The number reported in parenthesis is 1000 $\times$ rmse. The square bracket to the right of each estimator shows its theoretical limit, against which the numbers are computed.}
\end{footnotesize}
\end{center}
\end{table}

We first navigate through the left-hand portion of the table, which covers the results in the absence of microstructure noise. The presence of noise case is summarized towards the end of the section. First, as noticeable from the table, the realized range-based estimators are mildly biased. This was to be expected for the jump-robust measures, because $n$ is not sufficiently high in these simulations to fully eradicate the impact of jumps. Interestingly, though, the relative bias of $RBV^{n,m}$ and $RTV^{n,m}$ is a bit smaller that what we compute for the return-based competitors, illustrating that the jump-robust range-statistics appear less sensitive to this effect in small samples. Second, and by contrast to the subsampled realized variance, there is also a slight bias in the combined estimator $RRV^{n,m}$, but it is only about one percent at the 5-minute sampling frequency. In any instance, the rmse of the range-based estimators is much smaller compared to the equivalent return-based estimators, which reinforces the statements from the previous section, based on large sample theory. Thus, the range-statistic maintains also a comfortable lead in sampling stability in finite samples. Lastly, the table reaffirms that the range-based tripower extension is more efficient than its bipower companion. As readily seen, it is also less biased. Hence, going forward we restrict attention to $RTV^{n,m}$.

\begin{figure}[t!]
\begin{center}
\caption{Asymptotic approximation of the realized range-based tripower variance.
\label{Figure:simPdf}}
\begin{tabular}{cc}
\footnotesize{Panel A: $RTV^{n,m}$} & \footnotesize{Panel B: $\ln RTV^{n,m}$}\\
\includegraphics[height=8.00cm,width=0.45\linewidth]{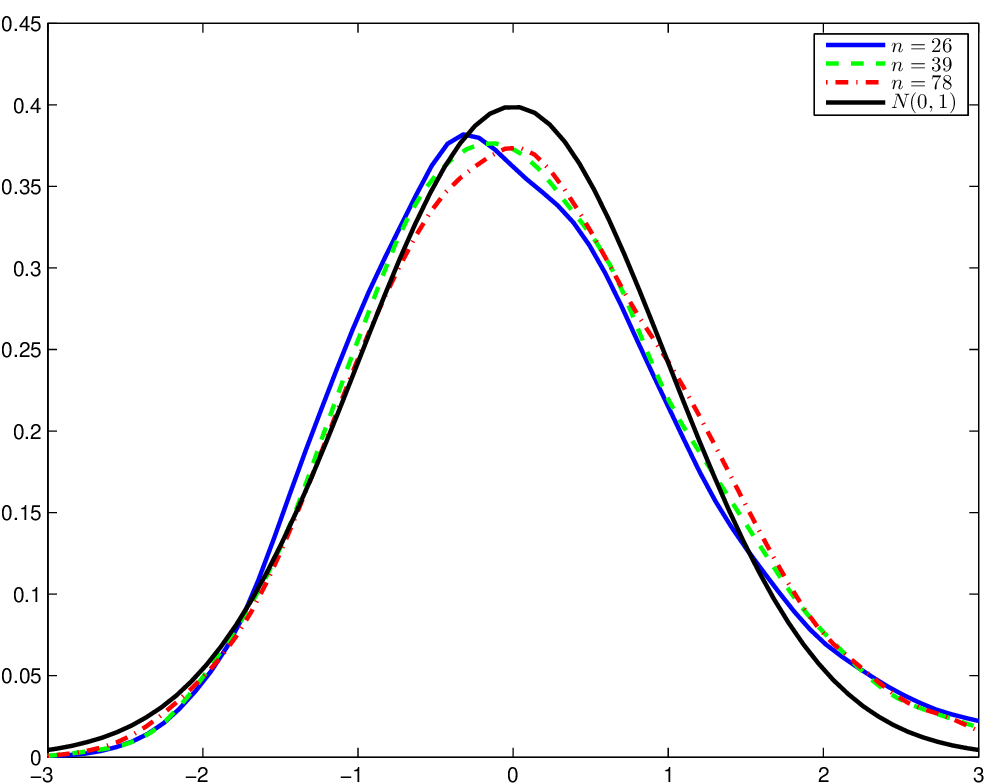} &
\includegraphics[height=8.00cm,width=0.45\linewidth]{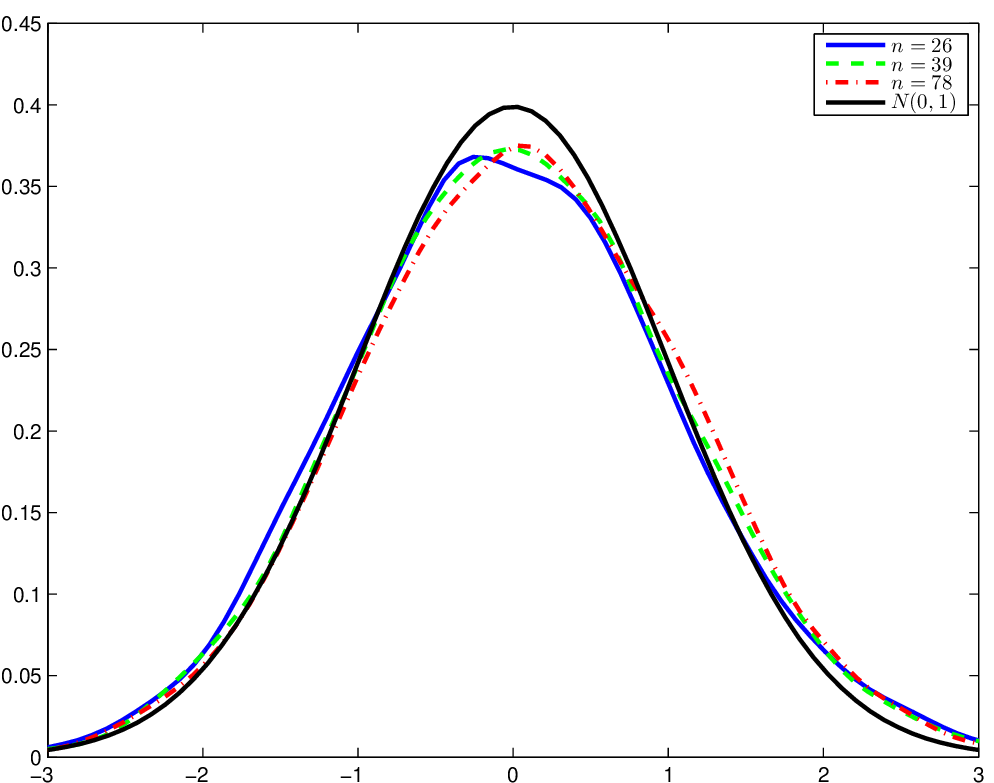} \\
\end{tabular}
\begin{scriptsize}
\parbox{0.925\textwidth}{\emph{Note.} We present smoothed density plots of the raw and log-based sampling distribution of the standardized realized range-based tripower variance estimator. The sampling frequency runs through $n = 26, 39, 78$, which represents 15-, 10- and 5-minute sampling and corresponds to $m = 90,60,30$.}
\end{scriptsize}
\end{center}
\end{figure}

To corroborate the analysis, we turn to Figure \ref{Figure:simPdf}, which provides a reality check on the accuracy of the limiting normal distribution derived for $RTV^{n,m}$. Note that Panel A is based on the standardized version of the CLT from Theorem \ref{Thm:RMV}, while Panel B uses the delta method to conclude that
\begin{equation}
\frac{\sqrt{n} \left( \ln RTV^{n,m} - \ln \int_{0}^{1} \sigma_{u}^{2} \text{d}u \right)}{\sqrt{\Lambda_{RTV^{n,m}}^{m} \int_{0}^{1} \sigma_{u}^{4} \text{d}u / \left( \int_{0}^{1} \sigma_{u}^{2} \text{d}u \right)^{2}}} \overset{d}{ \to} N(0,1).
\end{equation}
As apparent from Panel A, the sampling distribution of $RTV^{n,m}$ deviates from the standard normal with some distortions both in the tails and center area of the density, although the fit does gradually improve as $n$ rises. On the other hand, the log-based distribution theory in Panel B tracks the finite sample distribution of $\ln RTV^{n,m}$ somewhat better at all sampling frequencies. Thus, as a practical recommendation, we advocate using the log-based approximation, which has the added virtue of enforcing non-negativity on confidence bands for the integrated variance.

\begin{figure}[t!]
\begin{center}
\caption{Jump-robust 95\% confidence intervals for the integrated variance.
\label{Figure:simCI}}
\begin{tabular}{c}
\includegraphics[height=8.00cm,width=0.75\linewidth]{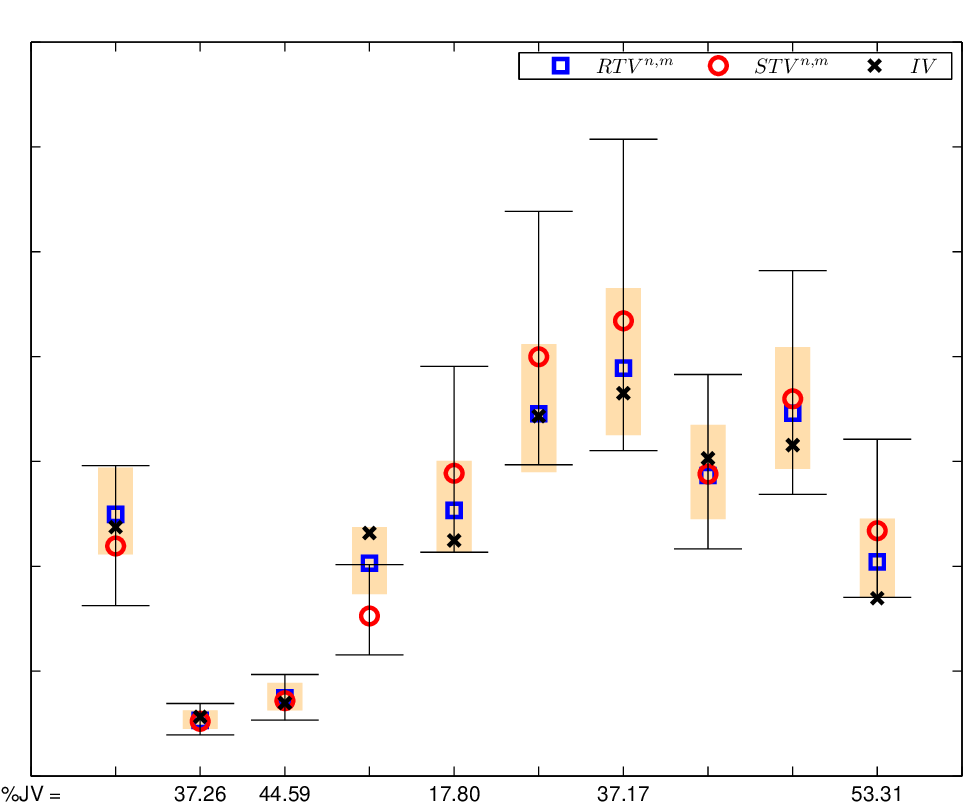}
\end{tabular}
\begin{scriptsize}
\parbox{0.925\textwidth}{\emph{Note.} This chart shows point estimates of $RTV^{n, m}$ and $STV^{n, m}$ and 95\% confidence intervals for the integrated variance, using $n = 78$. The box (``whisker'') is the range-based (return-based) confidence interval. In both cases, the integrated quarticity is proxied with a tripower estimator. \textbf{x} marks the true value of integrated variance. The labels on the x-axis report the proportion of quadratic variation induced by jumps (blank if zero).}
\end{scriptsize}
\end{center}
\end{figure}
Next, we consider the problem of making inference about diffusive return variation, as it would be done in practice, where the integrated variance is unknown. Figure \ref{Figure:simCI} shows 95\% confidence intervals for the integrated variance, using the log-based distribution theory and $n = 78$.\footnote{Note that because the log-based approximation is used, the confidence intervals are not completely symmetric around the point estimate.} As an illustration, we took a sequence of ten simulations, some of which include both small and large jumps. The x-axis labels report, in trials with jumps, how large is the jump as a ratio of total quadratic variation. In order to present a feasible theory, we replaced the unobserved integrated quarticity by a consistent jump-robust realized range-based multipower variation estimator. Among several candidates, we chose a tripower estimator for the job (with parameter $q_{j} = 4/3$, for $j = 1,\ldots, 3$). Also, as a benchmark in the chart, we compare with the return-based confidence intervals using $STV^{n, m}$, where the feasible bands are computed using a subsampled return-based tripower quarticity estimator.

Across all 10,000 simulations, the coverage rates are almost equal, whether using the return- or range-based approach. For example, at the 95\% level displayed in the figure, the range-based intervals include the integrated variance 93.38\% of the times, while this number is marginally better at 93.81\% for the return-based interval. Meanwhile, at the 99\% confidence level, the rates change to 98.56\% and 98.46\%, respectively, yielding almost identical performance. As evident, however, the range-based confidence intervals typically deliver sharper inference with much tighter bands, which is a reflection of the smaller asymptotic variance embedded in such estimators.

Finally, consider the right-hand part of Table \ref{Table:SimRMSE}, which discloses how the results are altered, when the efficient log-price data are concealed behind a realistic level of microstructure noise. As apparent, the range by and large preserves its relative efficiency advantage also under noise, at least for the moderate sampling frequencies considered in this paper. It is precisely in this setting, so often used in practice, we believe the range-based estimators hold some potential compared to the mainstream estimators. They are easy to implement, very efficient and fairly robust to noise at low-frequency. But, as consistent with the analysis of \citet*{christensen-podolskij-vetter:09a}, we also observe that the noise has a larger impact on the range vis-\`{a}-vis the return for a fixed value of $n$. Thus, if higher sampling frequencies are required, the noise will eventually swamp the advantages of the range and a noise-robust estimator should be adopted instead.

\section{Empirical application}
\label{Section:Empiri}
We illustrate some features of the range-based multipower variation theory for a few members of the Dow Jones Industrial Average index. The exposition is based on transaction data for the following three constituents: American Express (AXP), Merck (MRK) and Exxon (XOM). We also include data for the S\&P 500 Depository Receipts (SPY); an exchange-traded fund tracking the S\&P 500. We retrieved high-frequency data for these tickers from the TAQ database via the WRDS interface. The sample period spans the whole of 2007 -- 2009 and, as such, includes part of the ongoing financial crisis. Prior to analysis, we filtered the raw data for outliers, applying a set of rules that follow the guide proposed by \citet*{barndorff-nielsen-hansen-lunde-shephard:09a}.\footnote{In our application, we maintain the complete set of transaction data from every exchange, but we synchronize with quotes originating from the primary exchange only.} Moreover, we restrict attention to the regular trading hours and so remove updates with a timestamp outside 9:30{\scriptsize AM} to 4:00{\scriptsize PM} Eastern Standard Time.

Before we commence with the empirical analysis, it is worth to elaborate on a couple of practical points. Firstly, in real markets an important source of noise is price discreteness. Because of this feature -- but also other aspects of the market microstructure (e.g., the practice of splitting large block trades into smaller slices) -- it is not unusual to find long stretches in the data, where the price either repeats or only alternates between the bid and ask quotation (the so-called bid-ask bounce). Here, we recall that the range-based theory requests the number of ``price changes'' $m$ contained in each sampling interval, on which the range is being computed, as this serves as a prerequisite for returning the appropriate scaling factor $\lambda_{r,m}$. Hence, if we were to count every observation in the data as an increment to the price process, such practical features would tend to unduly propagate $m$. The outcome is an inflation of $\lambda_{r, m}$, which transforms into a downward bias in the range-based estimators.\footnote{The problem is in some sense akin to the no-trade bias of bipower variation, see \citet*{corsi-pirino-reno:10a}.} Meanwhile, designing a good algorithm to tally the ``true'' number of price changes requires us to deliver a formal definition of what we perceive to constitute an increment to the price process; a very subjective and challenging task.

Secondly, and equally important, the theory also calls upon an equidistant grid of log-price observations, while real trade arrivals are, of course, random. With irregularly spaced observations, the scaling factors $\lambda_{r, m}$, which are built from equidistant data, are no longer correct and using these can have a profound effect on the estimation. Consider, for example, a uniform sampling scheme from a Brownian motion, in which 2,341 log-price observations are drawn without replacement from the entire set of 23,401 possible one-second time stamps available in the 6.5 hours trading day, yielding $N = 2,340$ irregularly spaced returns. Suppose also, as in the simulation section, that we split the data into $n = 78$ intervals containing $m = 30$ returns each. Then, Table \ref{Table:EmpRMSE} illustrates the resulting relative bias and mse:

\begin{table}[H]
\setlength{\tabcolsep}{0.35cm}
\begin{center}
\caption{Relative bias and mean squared error with irregularly spaced data.
\label{Table:EmpRMSE}}
\smallskip
\begin{tabular}{lccccc}
\hline \hline
$(n,m) = (78,30)$ & \multicolumn{2}{c}{return-based} &&  \multicolumn{2}{c}{range-based}\\
\cline{2-3} \cline{5-6}
& $SRV^{n,m}$ & $STV^{n,m}$ && $RRV^{n,m}$ & $RTV^{n,m}$\\
Relative bias & 1.000 & 0.990 && 0.965 & 0.947\\
Mse$\times n$ & 1.399 & 1.665 && 0.643 & 0.681\\
\hline \hline
\end{tabular} \smallskip
\begin{footnotesize}
\parbox{0.925\textwidth}{\emph{Note}. We report the relative bias and mse of the estimators using irregularly spaced data. The bias measure is 1 for an unbiased estimator. The mse has been normalized by multiplying with $n$.}
\end{footnotesize}
\end{center}
\end{table}

As apparent, with irregular spacings the use of $\lambda_{r,m}$ leads to non-trivial downward biases in the range-based estimators, which reduces their advantage from an mse point of view. The solution to this problem is to tailor $\lambda_{r,m}$ to the observation times. This implies simulating scaling factors on the fly, which is less appealing. But, it can be done \citep*[e.g.][]{rossi-spazzini:09a}. We feel this issue is not too worrying for the data series under investigation here, which are comprised of deeply liquid securities.\footnote{For example, if we modify the above simulation exercise to draw 11,700 irregular observations of Brownian motion (matching closer with our empirical data) and then construct an artificial 10-second record of log-prices using previous-tick interpolation (as done below), the downward bias is reduced to a mere 0.2\% in both $RRV^{n,m}$ and $RTV^{n,m}$ across 10,000 trials.} However, when working with older data sets or illiquid series, this issue has more substance, and here it would be an advantage, and probably necessary, to suitably account for the irregular nature of high-frequency data by simulating grid-specific scalings.

As a consequence of the above features, i.e. price discreteness and irregular spacing,  a further modification of the theory is necessary for the empirical use of the realized range-based multipower variation framework. Therefore, we settled on the following compromise. Throughout the trading day, we collect a new observation of $p$ every 10th second, using previous-tick interpolation to replace missing values by the most recent transaction price. We then construct the realized range- and return-based estimators using 5-minute sampling (i.e., setting $n = 78$ and $m = 30$). Table \ref{Table:Descriptive} holds descriptive statistics of the data and resulting series.

\begin{table}[H]{
\renewcommand{\arraystretch}{0.85}
\setlength{\tabcolsep}{0.35cm}
\begin{center}
\caption{Descriptive statistics of the equity data.
\label{Table:Descriptive}}
\smallskip
\begin{tabular}{lrrrrrrrrrrrrrrrrrrr}
\hline \hline
& \multicolumn{4}{c}{\textbf{AXP}} && \multicolumn{4}{c}{\textbf{MRK}} \\
\cline{2-5} \cline{7-10}
                & 2007 & 2008 & 2009 & all  && 2007 & 2008 & 2009 & all \\
nobs(avg. 1000) &  8.3 & 12.3 & 11.2 & 10.6 &&  9.4 & 11.8 & 10.5 & 10.6 \\ \\
$RRV^{n,m}$     & 24.5 & 59.8 & 51.4 & 45.3 && 20.7 & 38.9 & 30.3 & 30.0 \\
$RTV^{n,m}$     & 21.8 & 55.4 & 47.2 & 41.5 && 18.1 & 34.6 & 26.7 & 26.5 \\
JV(range)       & 10.9 &  7.5 &  8.2 &  8.4 && 12.4 & 11.1 & 12.0 & 11.7 \\ \\
$SRV^{n,m}$     & 23.2 & 56.0 & 49.2 & 42.8 && 19.4 & 36.3 & 28.6 & 28.1 \\
$STV^{n,m}$     & 21.7 & 53.7 & 46.6 & 40.7 && 18.0 & 34.1 & 26.8 & 26.3 \\
JV(return)      &  6.5 &  4.2 &  5.3 &  5.1 &&  7.3 &  6.2 &  6.4 &  6.5 \\ \\
& \multicolumn{4}{c}{\textbf{SPY}} && \multicolumn{4}{c}{\textbf{XOM}} \\
\cline{2-5} \cline{7-10}
                & 2007 & 2008 & 2009 & all  && 2007 & 2008 & 2009 & all \\
nobs(avg. 1000) & 15.3 & 19.4 & 19.2 & 17.9 && 14.3 & 16.8 & 14.6 & 15.2 \\ \\
$RRV^{n,m}$     & 11.7 & 26.5 & 19.6 & 19.3 && 21.1 & 36.4 & 22.4 & 26.6 \\
$RTV^{n,m}$     & 11.0 & 25.6 & 18.5 & 18.4 && 19.9 & 34.4 & 20.9 & 25.1 \\
JV(range)       &  5.5 &  3.4 &  5.9 &  4.6 &&  5.6 &  5.5 &  6.6 &  5.8 \\ \\
$SRV^{n,m}$     & 11.5 & 26.0 & 19.2 & 18.9 && 20.2 & 34.9 & 21.5 & 25.6 \\
$STV^{n,m}$     & 10.9 & 25.2 & 18.1 & 18.1 && 19.5 & 33.8 & 20.5 & 24.6 \\
JV(return)      &  4.5 &  3.2 &  5.8 &  4.3 &&  3.8 &  3.2 &  4.4 &  3.7 \\
\hline \hline
\end{tabular}\smallskip
\begin{footnotesize}
\parbox{0.925\textwidth}{\emph{Note}. We report descriptive statistics for AXP, MRK, SPY and XOM. The sample period is January, 2007 -- December, 2009, both included. nobs is the average number of transaction data after cleaning (in 1000s). Volatility is reported as annualized standard deviation, in percent. JV is one minus the average ratio of the jump-robust integrated variance estimate to the quadratic variance estimate.}
\end{footnotesize}
\end{center}
}\end{table}

As can be gleaned from the table, the range-based multipower variations deliver estimates of quadratic return variation and integrated variance, which are in line or marginally above the corresponding return-based estimates. This holds across the entire sample or by yearly subsample periods, as reported in Table \ref{Table:Descriptive}. Browsing through the four equities, the range suggests that jump variation, i.e. the proportion of total variation produced by jumps, is in the order of 4.6\% -- 11.7\%, which broadly agrees with the return-based estimates and also the previous literature. This is comforting, as we would expect the unconditional sample averages of both return-based and range-based estimation to be broadly in line. It is also interesting to note that SPY, representing the S\&P 500 index, has the lowest estimated jump proportion.

\begin{figure}[t!]
\begin{center}
\caption{Illustration using MRK.
\label{Figure:TS[MRK]}}
\begin{tabular}{cc}
\footnotesize{Panel A: $RRV^{n,m}$, $RTV^{n,m}$ and jump variation} & \footnotesize{Panel B: MRK transaction data, 20080125}\\
\includegraphics[height=8.00cm,width=0.45\linewidth]{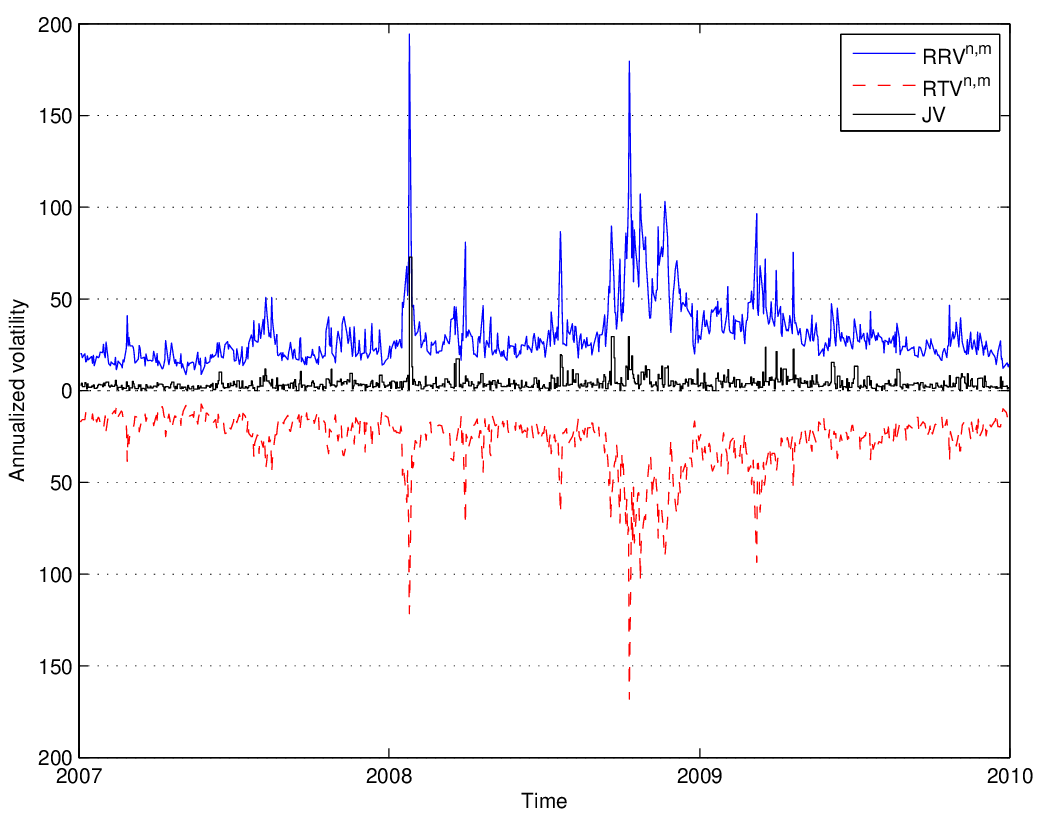} &
\includegraphics[height=8.00cm,width=0.45\linewidth]{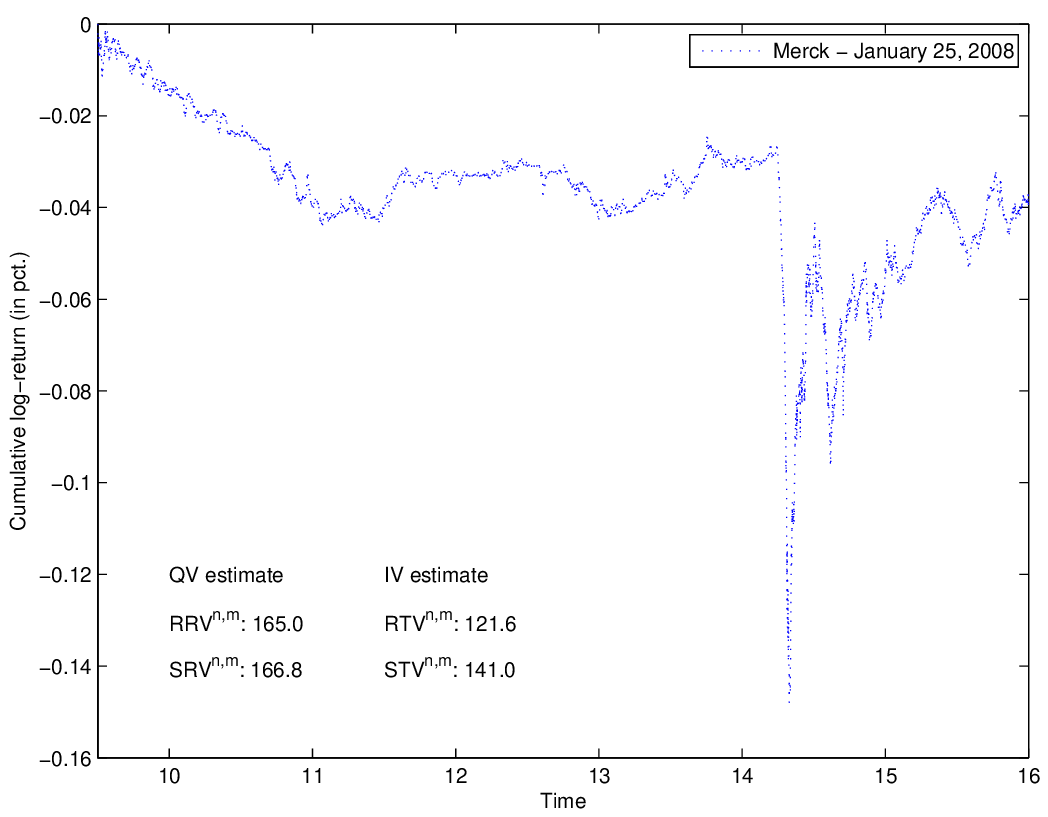} \\
\end{tabular}
\begin{scriptsize}
\parbox{0.925\textwidth}{\emph{Note.} In Panel A, we plot the complete time series of $RRV^{n,m}$ and $RTV^{n,m}$ for MRK. The latter has been reflected in the x-axis to improve the visual layout. Jump variation is the defined as $\max \left( RRV^{n,m} - RTV^{n,m}, 0 \right)$. The series have been converted to an annualized standard deviation measure. In Panel B, we plot the filtered transaction data for MRK on January 25, 2008.}
\end{scriptsize}
\end{center}
\end{figure}

In Panel A of Figure \ref{Figure:TS[MRK]}, we apply the range-statistic to estimate the overall level of return variation and its composition using the high-frequency data from MRK as an illustration. In the graph, all series are converted to an annualized standard deviation term. As readily observed, in general the $RRV^{n,m}$ and $RTV^{n,m}$ series tend to swing in parallel, although there a few notable departures, as revealed by the jump variation (JV) figure also reported in the chart. Panel B investigates such an instance by plotting the high-frequency data for MRK on January 25, 2008, where the JV measure is about 75\%. On this day, Merck published a press release regarding its anti-cholesterol product Mevacor, which received a ``not approvable'' letter from the FDA, while government regulators also said they were analyzing recent results from the clinical trials of the company's Vytorin drug (the so-called ENCHANCE study). The news were largely perceived as bad by the market, moving the equity into deep negative territory, although most of the losses were recovered before the end of trading. Of course, to reach statistical conclusions about the presence of jumps in the sample path, we would need to do a formal hypothesis test, but the plot and point estimates do indicative this.

\begin{figure}[t!]
\begin{center}
\caption{Comparison of return- and range-based estimator of the integrated variance, SPY data.
\label{Figure:ACF-CI[SPY]}}
\begin{tabular}{cc}
\footnotesize{Panel A: ACF of $RTV^{n,m}$ and $STV^{n,m}$} & \footnotesize{Panel B: Estimates of IV, September 2008}\\[-0.25cm]
\includegraphics[height=8.00cm,width=0.45\linewidth]{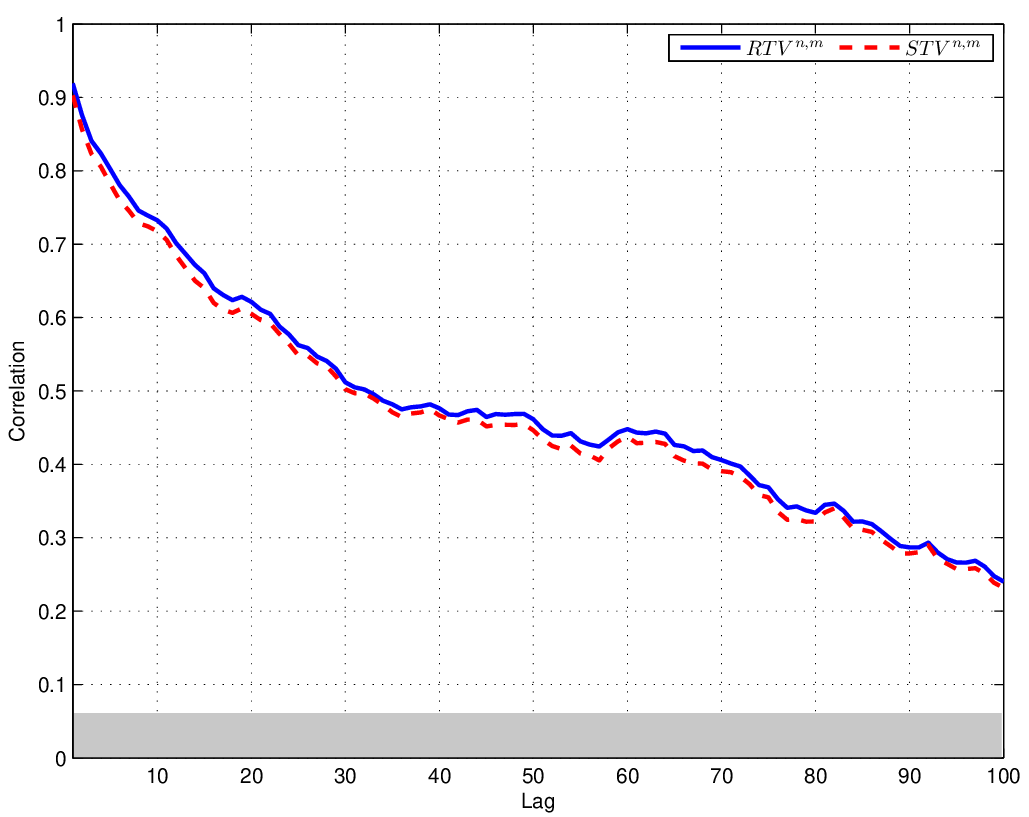} &
\includegraphics[height=8.20cm,width=0.45\linewidth]{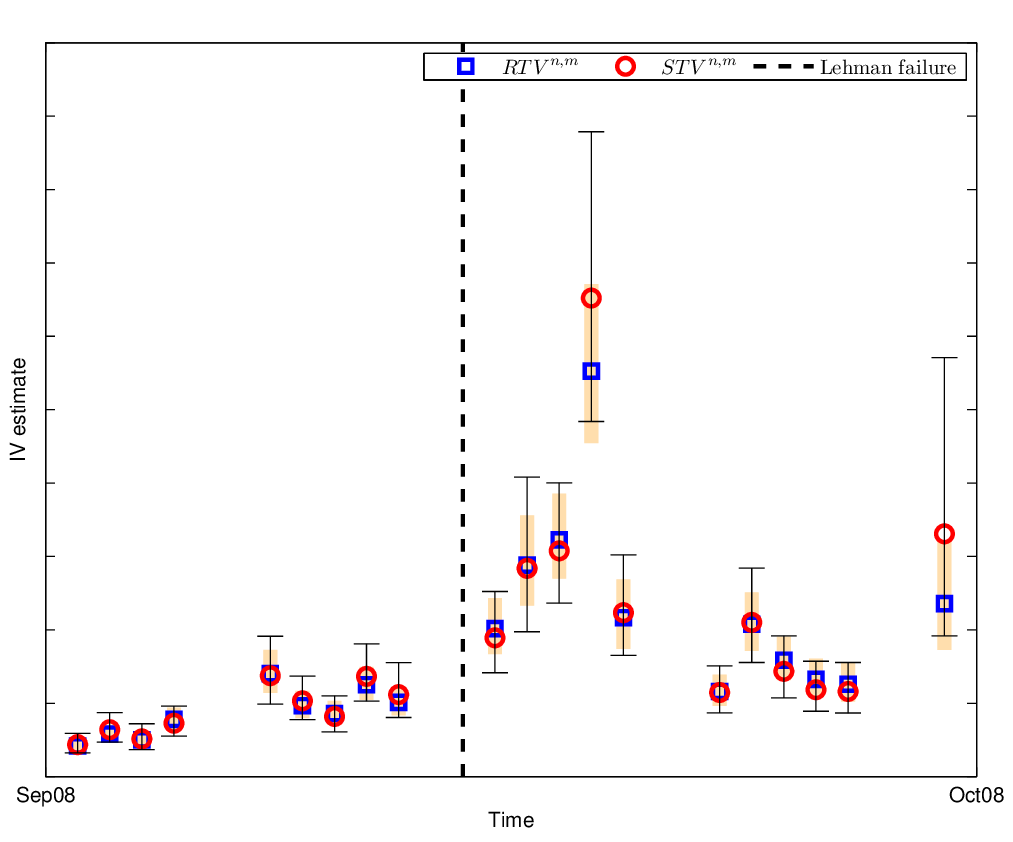} \\
\end{tabular}
\begin{scriptsize}
\parbox{0.925\textwidth}{\emph{Note.} To the left, we plot the autocorrelation function (ACF) of the jump-robust $RTV^{n,m}$ and $STV^{n,m}$ from lags 1 -- 100. The shaded area defines the Bartlett's two standard error bands for testing a white noise hypothesis. To the right, we display the $RTV^{n,m}$ and $STV^{n,m}$ point and interval estimates of the integrated variance for the month of September, 2008.}
\end{scriptsize}
\end{center}
\end{figure}

Finally, in Figure \ref{Figure:ACF-CI[SPY]}, we compare the range-based and subsampled return-based tripower estimators of the integrated variance. In the left panel, we plot the autocorrelation functions of the two estimators up to 100 lags. As evident, they are virtually identical and both display the typical long-range dependence that has been observed in volatility many times before \citep*[e.g.][]{andersen-bollerslev-diebold-labys:03a}. The range-statistic is slightly more persistent than the return-based companion, which is in fact true across all equities considered here (not reported). Although the differences between the two series are larger in magnitude for some of the other symbols, it is not clear if this can be exploited to produce superior forecasting accuracy. We leave this task for future work to decide. In the right panel, we mimic the approach taken in the simulation section to construct feasible confidence intervals for the integrated variance in the highly volatile month of September, 2008, which featured the demise of Lehman Brothers, the 4th largest investment bank in the U.S. at that time. Notice the widening of the error bands (via the estimate of integrated quarticity) as $\sigma$ increases, which reflects the surge in volatility in the aftermath of the bankruptcy. Again, and according to theory, the range-based intervals are smaller than the return-based ones, thus delivering sharper inference about the integrated variance.

\section{Conclusions and directions for future research}

This paper has presented the concept of realized range-based multipower variation and shown how it can be used to estimate the ex-post quadratic return variation and conduct jump-robust inference about integrated variance. The large sample asymptotic theory was backed by both a simulation study and an empirical application, illustrating the potential of the range-statistic. The range was motivated by the typical sparse sampling of return-based estimators caused by the presence of microstructure noise in financial high-frequency data. Of course, the range itself is affected by the noise component, as we demonstrated with numerical simulations. In a companion paper \citep*{christensen-podolskij-vetter:09a}, we study more formally the impact of noise on the standard realized range-based variance, and the interested reader is encouraged to read this material.

The theory developed here casts new light on the range and displays its properties in a general semimartingale model with stochastic volatility and finite activity jumps. But there are still several interesting issues left untouched. First, covariance risk is important in financial economics and range-based measures are notoriously difficult to extend to the multivariate setting. Range-based covariance estimation has been studied in \citet*{brandt-diebold:06a} and \citet*{bannouh-dijk-martens:09a}, although they make rather strong assumptions about the underlying process driving the evolution of asset prices over time. Moreover, those papers are based on a polarization identity, which in general does not guarantee a positive semi-definite covariance matrix estimate. In an ongoing paper, we are working on an extension of some of the concepts discussed here to multivariate processes, and we hope to be able to publish some results soon. Second, it could be worth to consider alternative available tools, which may refine the asymptotic approximations presented here, for example bootstrapping, Edgeworth expansions or Box-Cox transformations, as has been suggested by \citet*{ait-sahalia-mykland-zhang:11a, goncalves-meddahi:09a, goncalves-meddahi:11a} in the context of the realized variance. Finally, as the realized range-based variance of \citet*{christensen-podolskij:07a, martens-dijk:07a}, which in this paper was shown to be biased, has been found to produce good forecasts of future return variation \citep*[e.g.][]{patton-sheppard:09a}, we look forward to a more thorough analysis of the forecasting capabilities of realized range-based multipower variations introduced here, for example following the lines of \citet*{andersen-bollerslev-diebold:07a}.

\section{Acknowledgements}

This work was supported by CREATES, which is funded by the Danish National Research Foundation, and the Deutsche Forschungsgemeinschaft through grant SFB 475 "Reduction of complexity in multivariate data structures". Mark Podolskij also received funding from the Microstructure of Financial Markets in Europe (MicFinMa) network. Previous drafts of this article were circulated under the titles ``Asymptotic theory for range-based estimation of quadratic variation of discontinuous semimartingales'' and ``Range-based estimation of quadratic variation.'' Parts of the paper were prepared, while Kim Christensen was visiting University of California, San Diego (UCSD), Rady School of Management, whose hospitality is gratefully acknowledged. We are indebted to George Tauchen (co-editor) and an anonymous referee for their valuable and insightful feedback. In addition, we thank Allan Timmermann, Asger Lunde, Holger Dette, Morten Nielsen, Neil Shephard, Per Mykland, Roel Oomen, Rossen Valkanov as well as conference and seminar participants at the 2006 CIREQ conference on ``Realized Volatility'' in Montr\'{e}al, Canada, the ``61$^{\text{st}}$ European Meeting of the Econometric Society'' (ESEM) 2006, Austria, the ESF workshop on ``High Frequency Econometrics and the Analysis of Foreign Exchange Markets'' at Warwick Business School, United Kingdom, the ``International Conference on High Frequency Finance'' in Konstanz, Germany, the ``Statistical Methods for Dynamical Stochastic Models'' conference in Mainz, Germany, the ``Stochastics in Science, in Honor of Ole E. Barndorff-Nielsen'' conference in Guanajuato, Mexico, the Rady School of Management, UCSD, and University of Florence for comments and suggestions. The usual disclaimer applies.

\clearpage

\appendix
\section*{Appendix of proofs}

\section*{Proof of Theorem \ref{IcRrv}}
First, we decompose
\begin{equation*}
RRV_{b}^{n,m} = \frac{1}{ \lambda_{2,m}} \left( \sum_{i \in \Gamma_{n}} s_{p_{i \Delta, \Delta}, m}^{2} +
\sum_{i \in \Gamma_{n}^{c}} s_{p_{i \Delta, \Delta}, m}^{2} \right),
\end{equation*}
where $\Gamma_{n} = \left\{ 1 \leq i \leq n \mid \text{the process } p \text{ jumps on } \left[ (i - 1)/n, i/n \right] \right\}$. Because the jump
part of $p$ has finite activity, there are only finitely many jumps on $[0,1]$, so the first sum in the above decomposition is finite (almost surely).
In addition, with a probability converging towards $1$, there is at most one jump per interval $\left[ (i - 1)/n, i/n \right]$. Combined with the results of \citet*{christensen-podolskij:07a} (or Theorem 2 in this paper), this implies that
\begin{equation*}
\frac{1}{ \lambda_{2,m}} \sum_{i \in \Gamma_{n}} s_{p_{i \Delta, \Delta}, m}^{2} \overset{p}{ \to} \frac{1}{ \lambda_{2,m}} \sum_{i = 1}^{N_{1}} J_{i}^{2} \qquad \text{and} \qquad \frac{1}{ \lambda_{2, m}} \sum_{i \in \Gamma_{n}^{c}} s_{p_{i \Delta, \Delta}, m}^{2} \overset{p}{ \to} \int_{0}^{1} \sigma_{u}^{2} \text{d}u,
\end{equation*}
i.e.
\begin{equation*}
RRV_{b}^{n,m} \overset{p}{ \to} \int_{0}^{1} \sigma_{u}^{2} \text{d}u + \frac{1}{ \lambda_{2,m}} \sum_{i = 1}^{N_{1}} J_{i}^{2},
\end{equation*}
as asserted. \hfill $\blacksquare$

\section*{Proof of Theorem 2}
{\bf Preliminaries and some notation}\\
First, we note that as $t \mapsto \sigma_{t}$ is c\`{a}dl\`{a}g, all powers of $\sigma$ are locally integrable with respect to the Lebesgue measure, so that for any $t$ and $s > 0$, $\int_{0}^{t} |\sigma_{u}|^{s} \text{d}u < \infty$. Moreover, and without loss of generality, we will restrict the functions $\mu$, $\sigma$, $\mu^{ \prime}$, $\sigma^{ \prime}$, $v^{ \prime}$ and $\sigma^{-1}$ to be bounded \citep*[e.g.,][Section 3]{barndorff-nielsen-graversen-jacod-podolskij-shephard:06a}.

Next, if a process $X^{n}$ is of the form:
\begin{equation*}
X^{n} = \sum_{i = 1}^{n} \zeta_{i}^{n},
\end{equation*}
for an array $\left( \zeta_{i}^{n} \right)$ and $X^{n} \overset{p}{ \to} 0$, we say that $\left( \zeta_{i}^{n} \right)$ is asymptotically negligible (AN).

Also, in the below we employ generic constants, which are denoted by $C$ or $C_{p}$ (the latter notation is applied, when the constant depends on some external parameter $p$).

Moreover, we define
\begin{equation}
\label{Eqn:ApxRange}
\beta_{i,j}^n = \sqrt{n} |\sigma_{\frac{i-1}{n}}| s_{W_{(i+j-1)\Delta,\Delta }, m}, \qquad j = 1, \ldots, k.
\end{equation}
This construct is used to locally approximate the true (rescaled) range $\sqrt{n} s_{p_{(i + j - 1)\Delta,\Delta }, m}$. We suppress the dependence
of $\beta_{i,j}^n$ on $m$ for notational convenience. Finally, we also set
\begin{equation*}
g_{j}(x) = \frac{1}{\lambda_{q_j,m}} x^{q_j}, \qquad j = 1, \ldots, k,
\end{equation*}
and
\begin{equation*}
\rho_x(f) = \mathbb{E}[f(|x|s_{W, m})].
\end{equation*}
It should be noted that
\begin{equation*}
\rho_x(g_j) = |x|^{q_j}, \qquad j=1, \ldots, k.
\end{equation*}

\noindent {\bf Structure of the multipower variation proof}\\
Before we write down the proof of the realized range-based multipower variation estimator, we will briefly sketch the main steps and ideas behind it.

\bigskip

\noindent \textbf{(1)} First, we prove consistency of the estimator
\begin{equation*}
\widetilde{RMV}_{(q_{1}, \ldots, q_{k})}^{n,m} = \frac{1}{n} \sum_{i = 1}^{n - k + 1} \prod_{j = 1}^{k} \frac{ \bigl( \beta_{i,j}^{n} \bigr)^{q_j}}{ \lambda_{q_{j},m}}.
\end{equation*}
which is an approximation of $RMV_{(q_{1}, \ldots, q_{k})}^{n,m}$ based on the representation in Eq. \eqref{Eqn:ApxRange}. The trick is then to show that the error committed by using this approximation goes to $0$. That this is true follows exactly as in the return-based setting \citep*[e.g.,][Section 6]{barndorff-nielsen-graversen-jacod-podolskij-shephard:06a}.

\bigskip

\noindent \textbf{(2)} The proof of the central limit theorem also starts by figuring out the corresponding result
for the approximating sequence, i.e. we first deduce that
\begin{equation*}
\frac{1}{\sqrt{n}} \sum_{i = 1}^{n - k + 1} \left\{ \prod_{j = 1}^{k} \frac{ \bigl( \beta_{i,j}^{n} \bigr)^{q_j}}{ \lambda_{q_{j},m}} - \prod_{j = 1}^{k} \rho_{ \sigma_{ \frac{i - 1}{n}}} (g_j) \right\} \overset{d_{s}}{\to} \sqrt{ \Lambda_{(q_{1}, \ldots, q_{k})}^{m}} \int_{0}^{1} | \sigma_{u} |^{q_+} \text{\textup{d}} B_{u}.
\end{equation*}
The basic tool used to work out this assertion is Theorem IX 7.28 in \citet*{jacod-shiryaev:03a}. We note that the computation of the asymptotic
conditional variance is quite tedious and requires some very lengthy calculations.

\bigskip

\noindent \textbf{(3)}  In the next step, we justify the approximation made in step (2). Here, we assume without loss of generality that $k = 2$
and prove a CLT for the canonical process:
\begin{equation*}
\frac{1}{ \sqrt{n}} \sum_{i = 1}^{n} \biggl\{ g_1 \left( \sqrt{n} s_{p_{i \Delta, \Delta},m} \right) g_2 \left( \sqrt{n} s_{p_{ \left( i + 1
\right) \Delta, \Delta},m} \right) - \mathbb{E} \left[ g_1 \left( \sqrt{n} s_{p_{i \Delta, \Delta},m} \right) g_2 \left( \sqrt{n} s_{p_{\left( i + 1 \right) \Delta, \Delta},m} \right) \mid \mathcal{F}_{ \frac{i - 1}{n}} \right] \biggr\}
\end{equation*}
This part is in fact also shown as in the return-based case \citep*[e.g.,][Section 5]{barndorff-nielsen-graversen-jacod-podolskij-shephard:06a}.

\bigskip

\noindent \textbf{(4)} And, finally, the last step is to show that the process from step (3) and the original normalized statistic
in Theorem \ref{Thm:RMV} are asymptotically equivalent. This is the most complicated part. It is much more involved than with return-based estimation, because the supremum is not a smooth functional. On the other hand, as with returns, additional problems arise when powers $q_{j} \leq 1$ appear in the multipower
variation, because $f(x)=|x|^{p}$ is non-differentiable for $p \in (0,1]$.

\section{Proof of consistency}
Write
\begin{equation*}
\widetilde{RMV}_{(q_{1}, \ldots, q_{k})}^{n,m} = \sum_{i = 1}^{n - k + 1} \chi_{i}^{n}
\end{equation*}
with
\begin{equation*}
\chi_{i}^{n} = \frac{1}{n} \prod_{j = 1}^{k} \frac{ \bigl( \beta_{i,j}^{n} \bigr)^{q_j}}{ \lambda_{q_{j},m}}.
\end{equation*}
As explained above, this statistic serves to approximate the true estimator $RMV_{(q_{1}, \ldots, q_{k})}^{n,m}$. Now, as in \citet*[][Section 6]{barndorff-nielsen-graversen-jacod-podolskij-shephard:06a}, it holds that
\begin{equation}
\label{Eqn:deltaSeq}
RMV_{(q_{1}, \ldots, q_{k})}^{n,m} - \widetilde{RMV}_{(q_{1}, \ldots, q_{k})}^{n,m} \overset{p}{ \to} 0.
\end{equation}
On the other hand:
\begin{equation}
\label{Eqn:PlimCE}
\sum_{i = 1}^{n - k + 1} \mathbb{E} \Bigl[ \chi_{i}^{n} \mid \mathcal{F}_{ \frac{i - 1}{n}} \Bigr] = \frac{1}{n} \sum_{i = 1}^{n - k + 1}
| \sigma_{ \frac{i - 1}{n}} |^{q_{+}} \overset{p}{ \to} \int_{0}^{1} | \sigma_{u} |^{q_{+}} \text{d}u,
\end{equation}
and
\begin{equation}
\label{Eqn:ChiAN}
\sum_{i = 1}^{n - k + 1} \Bigl( \chi_{i}^{n} - \mathbb{E} \Bigl[ \chi_{i}^{n} \mid \mathcal{F}_{ \frac{i - 1}{n}} \Bigr] \Big)
\overset{p}{ \to} 0,
\end{equation}
because the above summands are conditionally $k$-independent. Putting Eq. \eqref{Eqn:deltaSeq} together with Eq. \eqref{Eqn:PlimCE},
and using that the sequence is Eq. \eqref{Eqn:ChiAN} is AN, we are able to conclude that
\begin{equation*}
RMV_{(q_{1}, \ldots, q_{k})}^{n,m} \overset{p}{ \to} \int_{0}^{1} | \sigma_{u} |^{q_{+}} \text{d}u.
\end{equation*}
This completes the proof of the range-based multipower variation consistency property. \hfill $\blacksquare$

\section{Proof of the central limit theorem}
We divide the proof into several steps, as detailed in the overview above.

\subsection{A central limit theorem for the approximation}
First, we introduce the quantity
\begin{equation*}
U^{\prime n,m} = \frac{1}{\sqrt{n}} \sum_{i = 1}^{n - k + 1} \left\{ \prod_{j = 1}^{k} \frac{ \bigl( \beta_{i,j}^{n} \bigr)^{q_j}}{ \lambda_{q_{j},m}} - \prod_{j = 1}^{k} \rho_{ \sigma_{ \frac{i - 1}{n}}} (g_j) \right \},
\end{equation*}
which is an approximation of $\sqrt{n} \left( RMV_{(q_{1}, \ldots, q_{k})}^{n,m} - \int_{0}^{1} |\sigma_{u}|^{q_{+}} \text{d}u \right)$.

\begin{lemma}
\label{Lemma:1}
Assume that $p$ follows the diffusion process defined in Eq. \eqref{BSM}. As $n \to \infty$, it holds that
\begin{equation*}
U^{\prime n,m} \overset{d_{s}}{\to} \sqrt{ \Lambda_{(q_{1}, \ldots, q_{k})}^{m}} \int_{0}^{1} | \sigma_{u} |^{q_{+}}
\text{\textup{d}}B_{u}.
\end{equation*}
\end{lemma}
\noindent Before giving the proof of the Lemma, note that the following estimates are true, when $p$ is a diffusion:
\begin{equation*}
\mathbb{E} \left[| \beta_{i,j}^{n}|^{q} \right] \leq C_{q}, \qquad j = 1, \ldots, k,
\end{equation*}
for all $q > 0$ and a fixed $m \in \mathbb N$. \\[0.25cm]
\textbf{Proof of Lemma \ref{Lemma:1}}: By shifting the indices, we form a decomposition
\begin{equation*}
U^{\prime n,m} = \sum_{i = j}^{n} \zeta_{i}^{n} + o_{p} (1)
\end{equation*}
with
\begin{equation*}
\zeta_{i}^{n} = \frac{1}{\sqrt{n}} \sum_{j = 1}^{k} \left\{ g_{j}( \beta_{i - j + 1,j}^{n}) - \rho_{ \sigma_{ \frac{i - j}{n}}} (g_j) \right\}
\prod_{l = 1}^{j - 1} g_{l}( \beta_{i - j + 1, l}^{n}) \prod_{l = j + 1}^{k} \rho_{ \sigma_{ \frac{i - j}{n}}} (g_l).
\end{equation*}
We remark that the term $g_{j}( \beta_{i - j + 1,j}^{n}) - \rho_{ \sigma_{ \frac{i - j}{n}}} (g_j)$ is measurable with respect to $\mathcal{F}_{\frac{i}{n}}$, while the other factors in the definition of $\zeta_{i}^{n}$ are $\mathcal{F}_{\frac{i - 1}{n}}$-measurable.

By Theorem IX 7.28 in \citet*{jacod-shiryaev:03a}, the proof of Lemma \ref{Lemma:1} will follow, if we can verify the following five conditions on the sequence $(\zeta_{i}^{n})$:
\begin{align*}
&\textbf{(A)} \quad \sum_{i = k}^{n} \mathbb{E} \left[ \zeta_{i}^{n} \mid \mathcal{F}_{ \frac{i - 1}{n}} \right] \overset{p}{ \to} 0, \qquad \qquad \qquad \qquad
\, \, \: \textbf{(B)} \quad \sum_{i = k}^{n} \mathbb{E} \left[ | \zeta_{i}^{n} |^{2} \mid \mathcal{F}_{ \frac{i - 1}{n}} \right] \overset{p}{\to}
\Lambda_{\left(q_{1}, \ldots, q_{k} \right)}^{m} \int_{0}^{1} |\sigma_{u}|^{2q_{+}} \text{d}u, \\[0.25cm]
&\textbf{(C)} \quad \sum_{i = k}^{n} \mathbb{E} \left[ \zeta_{i}^{n} \left( W_{ \frac{i}{n}} - W_{ \frac{i - 1}{n}} \right) \mid \mathcal{F}_{ \frac{i - 1}{n}} \right] \overset{p}{ \to} 0, \qquad
\textbf{(D)} \quad \sum_{i = k}^{n} \mathbb{E} \left[ | \zeta_{i}^{n}|^{4} \mid \mathcal{F}_{ \frac{i - 1}{n}} \right] \overset{p}{ \to} 0,
\end{align*}
\begin{align*}
\textbf{(E)} \quad \sum_{i = k}^{n} \mathbb{E} \left[ \zeta_{i}^{n} \left( N_{ \frac{i}{n}} - N_{ \frac{i - 1}{n}} \right) \mid \mathcal{F}_{ \frac{i - 1}{n}} \right] \overset{p}{ \to} 0,
\end{align*}
where the last condition must hold for any bounded martingale $N$ that is orthogonal to $W$ (i.e. with quadratic covariation $[N,W] = 0$).

As $\zeta_{i}^{n}$ is a martingale difference, it follows that
\begin{equation*}
\mathbb{E} \left[ \zeta_{i}^{n} \mid \mathcal{F}_{ \frac{i - 1}{n}} \right] = 0.
\end{equation*}
This verifies condition \textbf{(A)}. Next, because $\zeta_{i}^{n}$ is an even functional in $W$ and $W \eqschw -W$, we get that
\begin{equation*}
\mathbb{E} \left[ \zeta_{i}^{n} \left( W_{ \frac{i}{n}} - W_{ \frac{i - 1}{n}} \right) \mid \mathcal{F}_{ \frac{i - 1}{n}} \right] = 0,
\end{equation*}
which implies condition \textbf{(C)}. Moreover, we also deduce that
\begin{equation*}
\mathbb{E} \left[ | \zeta_{i}^{n} |^{4} \mid \mathcal{F}_{ \frac{i - 1}{n}} \right] \leq \frac{C}{n^{2}},
\end{equation*}
and this proves \textbf{(D)}. That condition \textbf{(E)} is true follows from the work of \citet*{christensen-podolskij:07a}.

So, we are left with the task of proving condition \textbf{(B)}, which requires a straightforward but somewhat tedious
calculation. We start out by defining
\begin{equation*}
\rho_{i - j,i - h}(f,g) = \int f \left(| \sigma_{ \frac{i - j - 1}{n}} | x \right) g \left(| \sigma_{ \frac{i - h - 1}{n}}| x \right) \delta_m(x) \text{d}x, \qquad 1 \leq j,h \leq k,
\end{equation*}
where $\delta_m$ denotes the density of $s_{W, m}$. We note that due to the continuity of $\sigma$
\begin{equation}
\label{Eqn:Est1}
\sup_{i \leq n} \sup_{1 \leq j, h \leq k}  | \rho_{i - j,i - h}(f,g) - \rho_{ \sigma_{ \frac{i - k}{n}}} (fg) | \overset{p}{ \to} 0,
\end{equation}
a result that is used in the computations below. By also setting
\begin{equation*}
\mu_{i,j}^{n} = \left\{ g_{j} \left( \beta_{i - j + 1, j}^{n} \right) - \rho_{ \sigma_{ \frac{i - j}{n}}} (g_j) \right\}
\prod_{l = 1}^{j - 1} g_{l} \left( \beta_{i - j + 1, l}^{n} \right) \prod_{l = j + 1}^{k} \rho_{ \sigma_{ \frac{i - j}{n}}} (g_l)
\end{equation*}
we find the identity
\begin{equation*}
\mathbb{E} \left[ | \zeta_{i}^{n} |^{2} \mid \mathcal{F}_{ \frac{i - 1}{n}} \right] = \sum_{j = 1}^{k} \mathbb{E} \left[ | \mu_{i,j}^{n}|^{2} \mid \mathcal{F}_{ \frac{i - 1}{n}} \right] + 2 \sum_{j = 1}^{k - 1} \sum_{h = 1}^{k - j} \mathbb{E} \left[ \mu_{i, h}^{n} \mu_{i,h + j}^{n} \mid \mathcal{F}_{ \frac{i - 1}{n}} \right].
\end{equation*}
From this, we deduce that
\begin{equation*}
\mathbb{E} \left[ | \mu_{i,j}^{n} |^{2} \mid \mathcal{F}_{ \frac{i - 1}{n}} \right] = \left\{ \rho_{ \sigma_{ \frac{i - j}{n}}} (g_{j}^{2}) - \rho_{ \sigma_{ \frac{i - j}{n}}}^{2} (g_j) \right\} \prod_{l = 1}^{j - 1} g_{l}^{2} \left( \beta_{i - j + 1, l}^{n} \right) \prod_{l = j + 1}^{k} \rho_{ \sigma_{ \frac{i - j}{n}}}^{2}(g_l)
\end{equation*}
and
\begin{align*}
\mathbb{E} \left[ \mu_{i, h}^{n} \mu_{i, h + j}^{n} \mid \mathcal{F}_{ \frac{i - 1}{n}} \right] &= \left\{ \rho_{i - h, i - h - j} (g_{h + j} g_{h}) - \rho_{ \sigma_{ \frac{i - h - j}{n}}} (g_{h + j}) \rho_{ \sigma_{ \frac{i - h}{n}}} (g_h) \right\} \prod_{l = 1}^{h - 1} g_{l} \left( \beta_{i - h + 1, l}^{n} \right) \\[0.25cm]
&\times \prod_{l = h + 1}^{k} \rho_{\sigma_{ \frac{i - j}{n}}} (g_{l}) \prod_{l = 1}^{h + j - 1} g_{l} \left( \beta_{i - h - j + 1,l}^{n} \right) \prod_{l = h + j + 1}^{k} \rho_{ \sigma_{ \frac{i - h - j}{n}}}(g_{l}).
\end{align*}
Note that
\begin{equation*}
\prod_{l = 1}^{h - 1} g_{l} \left( \beta_{i - h + 1, l}^{n} \right) \prod_{l = 1}^{h + j - 1} g_{l} \left( \beta_{i - h - j + 1, l}^{n} \right) = \prod_{l = 1}^{h - 1} g_{l} \left( \beta_{i - h + 1, l}^{n} \right) g_{l+j} \left( \beta_{i - h - j + 1, l + j}^{n} \right) \times \prod_{l = 1}^{j} g_{l} \left( \beta_{i - h - j + 1, l}^{n} \right).
\end{equation*}
We remark that $g_{l} \left( \beta_{i - h + 1, l}^{n} \right)$ and $g_{l + j} \left( \beta_{i - h - j + 1, l + j}^{n} \right)$
include the same increment of the Brownian motion $W$. Putting the pieces together and then calling upon Eq. \eqref{Eqn:Est1}, we conclude that
\begin{align*}
&\sum_{i = j}^{n} \mathbb{E} \left[ | \zeta_{i}^{n} |^{2} \mid \mathcal{F}_{ \frac{i - 1}{n}} \right] \overset{p}{ \to} \left\{
\left( \prod_{j = 1}^{k} \lambda_{q_{j}, m} \right)^{-2} \left( \sum_{j = 1}^{k} \left\{ \lambda_{2q_{j}, m} - \lambda_{q_{j}, m}^{2} \right\}
\prod_{l = 1}^{j - 1} \lambda_{2q_{l}, m} \prod_{l = j + 1}^{k} \lambda_{q_{l}, m}^{2}  \right.\right. \\[0.25cm]
&+ 2 \sum_{j = 1}^{k - 1} \sum_{h = 1}^{k - j} \left\{ \lambda_{q_{h + j} + q_{h}, m} - \lambda_{q_{h + j}, m} \lambda_{q_{h}, m} \right\}
\prod_{l = h + j + 1}^{k} \lambda_{q_{l}, m} \prod_{l = h + 1}^{k} \lambda_{q_{l}, m} \left.\left. \prod_{l = 1}^{h - 1} \lambda_{q_{l} + q_{l + j}, m}
\prod_{l = 1}^{j} \lambda_{q_{l}, m} \right) \right\} \int_{0}^{1} | \sigma_{u} |^{2q_{+}} \text{d}u.
\end{align*}
What is left is to show that the constant in front of $\int_{0}^{1} | \sigma_{u} |^{2q_{+}} \text{d}u$ equals $\Lambda_{\left( q_{1}, \ldots, q_{k} \right)}^{m}$. First, make the observation that
\begin{align*}
\sum_{j = 1}^{k} \left\{ \lambda_{2q_{j}, m} - \lambda_{q_{j}, m}^{2} \right\} \prod_{l = 1}^{j - 1} \lambda_{2q_{l}, m} \prod_{l = j + 1}^{k} \lambda_{q_{l}, m}^{2} &= \sum_{j = 1}^{k} \left( \prod_{l = 1}^{j} \lambda_{2q_{l}, m} \prod_{l = j + 1}^{k} \lambda_{q_{l}, m}^{2} - \prod_{l = 1}^{j - 1} \lambda_{2q_{l}, m} \prod_{l = j}^{k} \lambda_{q_{l}, m}^{2} \right) \\[0.25cm]
&= \prod_{l = 1}^{k} \lambda_{2q_{l}, m} - \prod_{l = 1}^{k} \lambda_{q_{l}, m}^{2},
\end{align*}
where we applied a telescopic sum argument. Using the notation
\begin{equation*}
a_{h, j} = \prod_{l = h + j + 1}^{k} \lambda_{q_{l},m} \prod_{l = h + 1}^{k} \lambda_{q_{l},m} \prod_{l = 1}^{h} \lambda_{q_{l} + q_{l + j}, m}
\prod_{l = 1}^{j} \lambda_{q_{l},m}
\end{equation*}
allows us to reduce this further
\begin{align*}
&2 \sum_{j = 1}^{k - 1} \sum_{h = 1}^{k - j} \left\{ \lambda_{q_{h + j} + q_{h}, m} - \lambda_{q_{h + j}, m} \lambda_{q_{h}, m} \right\}
\prod_{l = h + j + 1}^{k} \lambda_{q_{l}, m} \prod_{l = h + 1}^{k} \lambda_{q_{l}, m} \prod_{l = 1}^{h - 1} \lambda_{q_{l} + q_{l + j}, m}
\prod_{l = 1}^{j} \lambda_{q_{l}, m} \\[0.25cm]
&= 2 \sum_{j = 1}^{k - 1} \sum_{h = 1}^{k - j} \left(a_{h, j} - a_{h - 1, j} \right)= 2 \sum_{j = 1}^{k - 1} \left(a_{k - j, j} - a_{0, j} \right) \\[0.25cm]
&= -2(k - 1) \prod_{l = 1}^{k} \lambda_{q_{l}, m}^{2} + 2 \sum_{j = 1}^{k - 1} \prod_{l = k - j + 1}^{k} \lambda_{q_{l}, m} \prod_{l = 1}^{k - j} \lambda_{q_{l} + q_{l + j}, m} \prod_{l = 1}^{j} \lambda_{q_{l},m},
\end{align*}
which, upon collecting pieces, completes the proof. \hfill $\blacksquare$

\subsection{Justification of the approximation: I}
To keep ideas fixed, we concentrate here on the $k = 2$ setting (i.e., range-based bipower variation). The scheme of the proof does not change for a general $k$, but of course it does get slightly more complicated, although not by much.

First, we define the process
\begin{align*}
U ( g_{1}, g_{2} )^{n,m} &= \frac{1}{ \sqrt{n}} \sum_{i = 1}^{n} \biggl\{ g_{1} \left( \sqrt{n} s_{p_{i \Delta, \Delta},m} \right) g_{2} \left( \sqrt{n} s_{p_{ \left( i + 1 \right) \Delta, \Delta}, m} \right) \nonumber \\[0.25cm]
&- \mathbb{E} \left[ g_{1} \left( \sqrt{n} s_{p_{i \Delta, \Delta}, m} \right) g_{2} \left( \sqrt{n} s_{p_{ \left( i + 1 \right) \Delta, \Delta}, m} \right) \mid \mathcal{F}_{ \frac{i - 1}{n}} \right] \biggr\},
\end{align*}
Then, the convergence
\begin{equation*}
U( g_{1}, g_{2})^{n,m} - U^{ \prime n,m} \overset{p}{ \to} 0
\end{equation*}
follows as in the return-based world, see \citet*[][Section 5]{barndorff-nielsen-graversen-jacod-podolskij-shephard:06a}. As a consequence
\begin{equation*}
U(g_{1}, g_{2})^{n,m} \overset{d_{s}}{ \to} \sqrt{ \Lambda_{\left( q_{1},q_{2} \right)}^{m}} \int_{0}^{1} | \sigma_{u} |^{q_{+}}
\text{d}B_{u},
\end{equation*}
which completes this part. \hfill $\blacksquare$

\subsection{Justification of the approximation: II}
Again, and without loss of generality (but a pleasant loss of complexity), we assume that $k = 2$. In light of the previous steps, the CLT will follow, if we can prove the convergence
\begin{equation*}
\sqrt{n} \left( RMV_{ \left( q_{1}, q_{2} \right)}^{n, m} - \int_{0}^{1} |\sigma_{u}|^{q_{+}} \text{d}u \right) - U(g,h)^{n, m} \overset{p}{ \to} 0,
\end{equation*}
for a fixed $m \in \mathbb{N}$. We do this by proving that
\begin{equation*}
\zeta_{i}^{n, m} = \frac{1}{ \sqrt{n}} \mathbb{E} \left[ g_{1} ( \sqrt{n} s_{p_{i \Delta, \Delta},m}) g_{2} ( \sqrt{n} s_{p_{ \left( i + 1 \right) \Delta, \Delta}, m}) \mid \mathcal{F}_{ \frac{i - 1}{n}} \right] - \sqrt{n} \int_{ \frac{i - 1}{n}}^{ \frac{i}{n}} \rho_{ \sigma_{u}} \left( g_1 \right) \rho_{\sigma_{u}} \left( g_2 \right) \text{d}u,
\end{equation*}
is AN. To accomplish this, we split $\zeta_{i}^{n,m}$ into:
\begin{equation*}
\zeta_{i}^{n,m} = \zeta_{i}^{ \prime n,m} + \zeta_{i}^{ \prime \prime n},
\end{equation*}
where
\begin{align*}
&\zeta_{i}^{ \prime n,m} = \frac{1}{ \sqrt{n}} \left( \mathbb{E} \left[ g_{1} \left( \sqrt{n} s_{p_{i \Delta, \Delta}, m} \right) g_{2} \left( \sqrt{n} s_{p_{ \left( i + 1 \right) \Delta, \Delta},m} \right) \mid \mathcal{F}_{ \frac{i - 1}{n}} \right] - \mathbb{E} \left[ g_{1} \left( \beta_{i, 1}^{n} \right) \mid \mathcal{F}_{ \frac{i - 1}{n}} \right] \mathbb{E} \left[ g_{2} \left( \beta_{i, 2}^{n} \right) \mid \mathcal{F}_{ \frac{i - 1}{n}} \right] \right), \\[0.25cm]
&\zeta_{i}^{ \prime \prime n} = \sqrt{n} \int_{ \frac{i - 1}{n}}^{ \frac{i}{n}} \left( \rho_{ \sigma_{u}} \left( g_{1} \right) \rho_{ \sigma_{u}} \left( g_{2} \right) - \rho_{ \sigma_{ \frac{i - 1}{n}}} \left( g_{1} \right) \rho_{ \sigma_{ \frac{i - 1}{n}}} \left( g_2 \right) \right) \text{d}u.
\end{align*}
It follows from \citet*[][Section 8]{barndorff-nielsen-graversen-jacod-podolskij-shephard:06a} that $\zeta_{i}^{ \prime \prime n}$ is AN. So, the only missing piece is to show the sequence $\zeta_{i}^{ \prime n,m}$ is also AN and we are done. We set
\begin{equation*}
\xi_{i}^{n, m} = \sqrt{n} s_{p_{i \Delta, \Delta}, m} - \beta_{i, 1}^{m}, \qquad
\xi_{i}^{ \prime n, m} = \sqrt{n} s_{p_{(i + 1) \Delta, \Delta}, m} - \beta_{i, 2}^{m}.
\end{equation*}
Using assumption (V$_{2}$), we introduce the random variables:
\begin{align}
\zeta \left( 1 \right)_{i}^{n, m} &= \sqrt{n} \, \underset{s, t \in I_{i,n}^m}{ \max \biggl( \int_{s}^{t} \sigma_{ \frac{i - 1}{n}} \text{d} } W_{u} + \int_{s}^{t} \mu_{ \frac{i - 1}{n}} \text{d}u \nonumber \\[0.25cm]
&+ \label{4.35} \int_{s}^{t} \left\{ \sigma_{ \frac{i - 1}{n}}^{ \prime} \left( W_{u} - W_{ \frac{i - 1}{n}} \right) + v_{ \frac{i - 1}{n}}^{ \prime} \left( B_{u}^{ \prime} - B_{ \frac{i - 1}{n}}^{ \prime} \right) \right\} \text{d}W_{u} \biggr) - \beta_{i,1}^{n}, \\[0.25cm]
\zeta \left( 2 \right)_{i}^{n,m} &= \sqrt{n} \biggl\{ \underset{ s, t\in I_{i,n}^m}{ \max \Bigl( \int_{s}^{t} \mu_{u} \text{d}u } + \int_{s}^{t} \sigma_{u} \text{d}W_{u} \Bigr) - \underset{s, t \in I_{i,n}^m}{ \max \biggl( \int_{s}^{t} \sigma_{ \frac{i - 1}{n}} } \text{d}W_{u} + \int_{s}^{t} \mu_{ \frac{i - 1}{n}} \text{d}u \nonumber \\[0.25cm]
\label{4.36} &+ \int_{s}^{t} \left\{ \sigma_{ \frac{i - 1}{n}}^{ \prime} \left( W_{u} - W_{ \frac{i - 1}{n}} \right) + v_{ \frac{i - 1}{n}}^{ \prime} \left( B_{u}^{ \prime} - B_{ \frac{i - 1}{n}}^{ \prime} \right) \right\} \text{d}W_{u} \biggr) \biggr\},
\end{align}
where $I_{i, n}^{m} = \left\{ t \mid t = \frac{i - 1}{n} + \frac{j}{N} \text{ for some } 0 \leq j \leq m \right\}$.
We get that
\begin{equation*}
\xi_{i}^{n, m} = \zeta \left( 1 \right)_{i}^{n, m} + \zeta \left( 2 \right)_{i}^{n, m},
\end{equation*}
with an identical decomposition holding for $\xi_{i}^{ \prime n, m}$. As in \citet*[][Section 7.1]{barndorff-nielsen-graversen-jacod-podolskij-shephard:06a} we obtain, under assumption (V$_{2}$), the estimate
\begin{equation*}
\mathbb{E} \left[ | \xi_{i}^{n, m} |^{q} \right] \leq C n^{- \frac{q}{2}},
\end{equation*}
for any $q > 0$ and uniformly in $i$. We rewrite $\zeta_{i}^{ \prime n, m}$ as follows
\begin{equation*}
\zeta_{i}^{ \prime n, m} = \mathbb{E} \left[ \delta_{i}^{n, m} \mid \mathcal{F}_{ \frac{i - 1}{n}} \right],
\end{equation*}
with $\delta_{i}^{n, m}$ defined by:
\begin{equation*}
\delta_{i}^{n,m} = \frac{1}{ \sqrt{n}} \left( g_1 \left( \sqrt{n} s_{p_{i \Delta, \Delta}, m} \right) g_2 \left( \sqrt{n} s_{p_{ \left( i + 1 \right) \Delta, \Delta}, m} \right) - g_1 \left( \beta_{i,1}^{n} \right) g_2 \left( \beta_{i,2}^{ n} \right) \right).
\end{equation*}
Next, we observe that
\begin{align*}
\delta_{i}^{n, m} &= \frac{1}{ \sqrt{n}} g_{1} \left( \sqrt{n} s_{p_{i \Delta, \Delta}, m} \right) \left( g_{2} \left( \sqrt{n} s_{p_{ \left( i + 1 \right) \Delta, \Delta},m} \right) - g_{2} \left( \beta_{i, 2}^{n} \right) \right) \\[0.25cm]
&+ \frac{1}{ \sqrt{n}} \left( g_{1} \left( \sqrt{n} s_{p_{i \Delta, \Delta}, m} \right) - g \left( \beta_{i, 1}^{n} \right) \right) g_{2} \left( \beta_{i, 2}^{n} \right) \equiv \delta_{i}^{ \prime n,m} + \delta_{i}^{ \prime \prime n,m}.
\end{align*}
We show that
\begin{equation*}
\mathbb{E} \left[ \delta_{i}^{ \prime \prime n, m} \mid \mathcal{F}_{ \frac{i - 1}{n}} \right]
\end{equation*}
is AN, but omit the proof for $\mathbb{E} \left[ \delta_{i}^{ \prime n,m} \mid \mathcal{F}_{ \frac{i - 1}{n}} \right] $ to save space. Note that
\begin{equation*}
\sum_{i = 1}^{n} \mathbb{E} \left[ \delta_{i}^{ \prime \prime n, m} \mid \mathcal{F}_{ \frac{i - 1}{n}} \right] = \sum_{i = 1}^{n} \mathbb{E} \left[
\vartheta_{i}^{n, m} \mid \mathcal{F}_{ \frac{i - 1}{n}} \right] + o_{p}(1)
\end{equation*}
with
\begin{equation*}
\vartheta_{i}^{n, m} = \frac{1}{ \sqrt{n}} g_{2} \left( \beta_{i, 2}^{n} \right) \nabla g_{1} \left( \beta_{i, 1}^{n} \right) \xi_{i}^{n, m}.
\end{equation*}
While the above approximation appears to be a simple application of the mean value theorem, this result is actually highly non-trivial, because the function $g_{1}$ is not differentiable at 0, when $q_{1} < 1$. Nonetheless, $\nabla g_{1} \bigl( \beta_{i, 1}^{n} \bigr)$ is well-defined a.s. and $\mathbb{E} \left[ \nabla g_{1} \bigl( \beta_{i, 1}^{n} \bigr) \right] < \infty$ under assumption (V$_1$). The approximation can be shown as in \citet*[][Section 8]{barndorff-nielsen-graversen-jacod-podolskij-shephard:06a}.

Recall that
\begin{equation*}
\xi_{i}^{n, m} = \zeta \left( 1 \right)_{i}^{n, m} + \zeta \left( 2
\right)_{i}^{n, m},
\end{equation*}
with $\zeta \left( 1 \right)_{i}^{n, m}$ and $\zeta \left( 2 \right)_{i}^{n,m}$ defined by \eqref{4.35} and \eqref{4.36},
respectively. Set
\begin{align*}
f_{in} \left( s, t \right) &= \sqrt{n} \sigma_{ \frac{i - 1}{n}} \left( W_{t} - W_{s} \right), \\[0.25cm]
g_{in} \left( s, t \right) &= n \int_{s}^{t} \mu_{ \frac{i - 1}{n}} \text{d}u + n \int_{s}^{t} \left\{ \sigma_{ \frac{i - 1}{n}}^{
\prime} \left( W_{u} - W_{ \frac{i - 1}{n}} \right) + v_{ \frac{i - 1}{n}}^{ \prime} \left( B_{u}^{ \prime} - B_{ \frac{i - 1}{n}}^{
\prime} \right) \right\} \text{d}W_{u} \\[0.25cm]
&= \mu_{ \frac{i - 1}{n}} g_{in}^{1} \left( s, t \right) + \sigma_{ \frac{i - 1}{n}}^{ \prime} g_{in}^{2} \left( s, t \right) + v_{
\frac{i - 1}{n}}^{ \prime} g_{in}^{3} \left( s, t \right),
\end{align*}
to achieve the identity:
\begin{equation*}
\zeta \left( 1 \right)_{i}^{n, m} = \underset{s, t \in I_{i, n}^{m}}{ \max \biggl( f_{in} \left( t, s \right) } + \frac{1}{
\sqrt{n}} g_{in} \left( t, s \right) \biggr) - \underset{ s, t \in I_{i, n}^{m}}{ \max f_{in} \left( t, s \right) }.
\end{equation*}
Imposing assumption (V$_{1}$):
\begin{align*}
\left( t_{in}^{*m} \left( W \right), s_{in}^{*m} \left( W \right) \right) &= \underset{s, t \in I_{i, n}^{m}}{ \arg \max f_{in} \left(
s, t \right)} \\[0.25cm]
&= \underset{s, t \in I_{i, n}^{m}}{ \arg \max \sqrt{n} ( W_{t} }  - W_{s} ) \\[0.25cm]
& \overset{d}{=} \underset{s, t = 0, 1, \ldots, m}{ \arg \max ( W_{t / m} -} W_{s / m} ).
\end{align*}
A standard result then states that the pair $\left( t_{in}^{*m} \left( W \right), s_{in}^{*m} \left( W \right) \right)$ is unique, almost
surely \citep*[e.g.][p. 107]{karatzas-shreve:91a}. Next, the following results, which are proven in \citet*{christensen-podolskij:07a}, present a useful stochastic expansion for $\zeta \left( 1 \right)_{i}^{n, m}$ and a result on $\zeta \left( 2 \right)_{i}^{n, m}$.
\begin{lemma}
\label{derivstoch} Given assumption (V$_{ \mathit{1}}$)
\begin{equation*}
\zeta \left( 1 \right)_{i}^{n, m} = \frac{1}{ \sqrt{n}} \Bigl\{ g_{in} \left( t_{in}^{*m} \left( W \right), s_{in}^{*m} \left( W
\right) \right) + \tilde{g}_{in}^{m} \Bigr\},
\end{equation*}
where
\begin{equation}
\label{4.40} \mathbb{E} \bigl[ | \tilde{g}_{in}^m |^{p} \bigr] =
\text{\textup{o}} \left( 1 \right),
\end{equation}
for all $p > 0$ and uniformly in $i$.
\end{lemma}

\begin{lemma}
\label{lem4.13} If $q \geq 2$, it then holds that
\begin{equation*}
\frac{1}{ \sqrt{n}} \sum_{i = 1}^{n} \Big( \mathbb{E} \left[ | \zeta \left( 2 \right)_{i}^{n, m} |^{q} \right] \Big)^{ \frac{1}{q}} \to 0,
\end{equation*}
for all $t > 0$.
\end{lemma}
\noindent Note that $\left( t_{in}^{*m} \left( W \right), s_{in}^{*m} \left( W \right) \right) = \left( s_{in}^{*m} \left( - W \right), t_{in}^{*m} \left(
- W \right) \right).$ Moreover, as $\left( W, B^{ \prime} \right) \overset{d}{=} - \left( W, B^{
\prime} \right)$ and $\nabla g_1 \bigl( \beta_{i,1}^{n} \bigr)$ is an even functional of $W$:
\begin{equation*}
\mathbb{E} \left[ \nabla g_1 \left( \beta_{i,1}^{n} \right) g_{in}^k
\left( t_{in}^{*m} \left( W \right), s_{in}^{*m} \left( W \right)
\right) \mid \mathcal{F}_{ \frac{i - 1}{n}} \right] = 0,
\end{equation*}
for $k = 1, 2, 3$. Hence
\begin{equation}
\label{4.42} \mathbb{E} \left[ \nabla g_1 \left( \beta_{i,1}^{n} \right)
g_{in} \left( t_{in}^{*m} \left( W \right), s_{in}^{*m} \left( W
\right) \right) \mid \mathcal{F}_{ \frac{i - 1}{n}} \right] = 0.
\end{equation}
Using H\"{o}lder's inequality, it follows that
\begin{align*}
| \, \mathbb{E} \left[ \vartheta_{i}^{n,m} \mid \mathcal{F}_{ \frac{i - 1}{n}} \right] | &= \frac{1}{ \sqrt{n}} \rho_{ \sigma_{ \frac{i - 1}{n}}} \left( g_2 \right) | \, \mathbb{E} \left[ \nabla g_1 \left( \beta_{i,1}^{n} \right) \left( \zeta \left( 1 \right)_{i}^{n,m} + \zeta \left( 2 \right)_{i}^{n,m} \right) \mid \mathcal{F}_{ \frac{i - 1}{n}} \right] | \\[0.25cm]
& \leq \frac{1}{ \sqrt{n}} \rho_{ \sigma_{ \frac{i - 1}{n}}} \left( g_2 \right) \biggl( | \, \mathbb{E} \left[ \nabla g_1 \left( \beta_{i,1}^{n} \right) \zeta \left( 1 \right)_{i}^{n,m} \mid
\mathcal{F}_{ \frac{i - 1}{n}} \right] | \\[0.25cm]
&+ \left( \mathbb{E} \left[ | \nabla g_1 \left( \beta_{i,1}^{n} | \right)^{p} \right] \right)^{ \frac{1}{p}} \left( \mathbb{E} \left[ | \zeta \left( 2 \right)_{i}^{n,m} |^{q} \right] \right)^{ \frac{1}{q}} \biggr),
\end{align*}
for some $p > 1, q \geq 2$ with $\left( q_{1} - 1 \right ) p > -1$ and $1 / p + 1 / q = 1$. Note that $\mathbb{E} \left[ \left( \nabla g_{1} \bigl( \beta_{i, 1}^{n} \bigr) \right)^{p} \right] \leq C_{p} < \infty$, when $\left( q_1 - 1 \right ) p > -1$. Finally, by combining Lemma \ref{derivstoch} and \ref{lem4.13} with Eq. \eqref{4.40} and \eqref{4.42}, we get the AN property of the sequence $\mathbb{E} \left[ \vartheta_{i}^{n,m} \mid \mathcal{F}_{ \frac{i - 1}{n}} \right]$. Hence, the CLT has been proven. \hfill $\blacksquare$

\clearpage

\renewcommand{\baselinestretch}{1.0}
\small
\bibliographystyle{rfs}
\bibliography{userref}

@ARTICLE{ait-sahalia-jacod:09a,
 AUTHOR = {Y. A\"{i}t-Sahalia and J. Jacod},
 YEAR = {2009},
 TITLE = {{Estimating the degree of activity of jumps in high frequency data}},
 JOURNAL = {Annals of Statistics},
 VOLUME = {37},
 NUMBER = {5},
 PAGES = {2202--2244}
}

@ARTICLE{ait-sahalia-jacod:09b,
 AUTHOR = {Y. A\"{i}t-Sahalia and J. Jacod},
 YEAR = {2009},
 TITLE = {{Testing for jumps in a discretely observed process}},
 JOURNAL = {Annals of Statistics},
 VOLUME = {37},
 NUMBER = {1},
 PAGES = {184--222}
}

@ARTICLE{ait-sahalia-jacod:11a,
 AUTHOR = {Y. A\"{i}t-Sahalia and J. Jacod},
 YEAR = {2011},
 TITLE = {Testing whether jumps have finite or infinite activity},
 JOURNAL = {Annals of Statistics},
 VOLUME = {39},
 NUMBER = {3},
 PAGES = {1689--1719}
}

@ARTICLE{ait-sahalia-mykland-zhang:11a,
 AUTHOR = {Y. A\"{i}t-Sahalia and P. A. Mykland and L. Zhang},
 YEAR = {2011},
 TITLE = {{Edgeworth expansions for realized volatility and related estimators}},
 JOURNAL = {Journal of Econometrics},
 VOLUME = {160},
 NUMBER = {1},
 PAGES = {190--203}
}

@ARTICLE{alizadeh-brandt-diebold:02a,
 AUTHOR = {S. Alizadeh and M. W. Brandt and F. X. Diebold},
 YEAR = {2002},
 TITLE = {{Range-based estimation of stochastic volatility models}},
 JOURNAL = {Journal of Finance},
 VOLUME = {57},
 NUMBER = {3},
 PAGES = {1047--1092}
}

@ARTICLE{andersen-benzoni-lund:02a,
 AUTHOR = {T. G. Andersen and L. Benzoni and J. Lund},
 YEAR = {2002},
 TITLE = {{An empirical investigation of continuous-time equity return models}},
 JOURNAL = {Journal of Finance},
 VOLUME = {57},
 NUMBER = {4},
 PAGES = {1239--1284}
}

@ARTICLE{andersen-bollerslev:97b,
 AUTHOR = {T. G. Andersen and T. Bollerslev},
 YEAR = {1997},
 TITLE = {{Intraday periodicity and volatility persistence in financial markets}},
 JOURNAL = {Journal of Empirical Finance},
 VOLUME = {4},
 NUMBER = {2},
 PAGES = {115--158}
}

@ARTICLE{andersen-bollerslev:98a,
 AUTHOR = {T. G. Andersen and T. Bollerslev},
 YEAR = {1998},
 TITLE = {{Answering the skeptics: Yes, standard volatility models do provide accurate forecasts}},
 JOURNAL = {International Economic Review},
 VOLUME = {39},
 NUMBER = {4},
 PAGES = {885--905}
}

@INCOLLECTION{andersen-bollerslev-diebold:10a,
 AUTHOR = {T. G. Andersen and T. Bollerslev and F. X. Diebold},
 YEAR = {2010},
 TITLE = {{Parametric and nonparametric volatility measurement}},
 BOOKTITLE = {Handbook of Financial Econometrics},
 EDITOR = {L. P. Hansen and Y. Ait-Sahalia},
 PAGES = {67--138},
 PUBLISHER = {North-Holland},
 ADDRESS = {Amsterdam}
}

@ARTICLE{andersen-bollerslev-diebold:07a,
 AUTHOR = {T. G. Andersen and T. Bollerslev and F. X. Diebold},
 YEAR = {2007},
 TITLE = {{Roughing it up: Including jump components in the measurement, modeling and forecasting of return volatility}},
 JOURNAL = {Review of Economics and Statistics},
 VOLUME = {89},
 NUMBER = {4},
 PAGES = {701--720}
}

@ARTICLE{andersen-bollerslev-diebold-labys:03a,
 AUTHOR = {T. G. Andersen and T. Bollerslev and F. X. Diebold and P. Labys},
 YEAR = {2003},
 TITLE = {{Modeling and forecasting realized volatility}},
 JOURNAL = {Econometrica},
 VOLUME = {71},
 NUMBER = {2},
 PAGES = {579--625}
}

@ARTICLE{andersen-bollerslev-huang:11a,
 AUTHOR = {T. G. Andersen and T. Bollerslev and X. Huang},
 YEAR = {2011},
 TITLE = {{A reduced form framework for modeling volatility of speculative prices based on realized variation measures}},
 JOURNAL = {Journal of Econometrics},
 VOLUME = {160},
 NUMBER = {1},
 PAGES = {176--189}
}

@MISC{andersen-dobrev-schaumburg:08a,
 AUTHOR = {T. G. Andersen and D. Dobrev and E. Schaumburg},
 YEAR = {2008},
 TITLE = {{Duration-based volatility estimation}},
 NOTE = {Working paper}
}

@ARTICLE{back:91a,
 AUTHOR = {K. Back},
 YEAR = {1991},
 TITLE = {{Asset prices for general processes}},
 JOURNAL = {Journal of Mathematical Economics},
 VOLUME = {20},
 NUMBER = {4},
 PAGES = {371--395}
}

@ARTICLE{bannouh-dijk-martens:09a,
 AUTHOR = {K. Bannouh and D. van Dijk and M. Martens},
 YEAR = {2009},
 TITLE = {{Range-based covariance estimation using high-frequency data: The realized co-range}},
 JOURNAL = {Journal of Financial Econometrics},
 VOLUME = {7},
 NUMBER = {4},
 PAGES = {341--372}
}

@INCOLLECTION{barndorff-nielsen-graversen-jacod-podolskij-shephard:06a,
 AUTHOR = {O. E. Barndorff-Nielsen and S. E. Graversen and J. Jacod and M. Podolskij and N. Shephard},
 YEAR = {2006},
 TITLE = {{A central limit theorem for realized power and bipower variations of continuous semimartingales}},
 BOOKTITLE = {From Stochastic Calculus to Mathematical Finance: The Shiryaev Festschrift},
 EDITOR = {Y. Kabanov and R. Lipster and J. Stoyanov},
 PAGES = {33--68},
 PUBLISHER = {Springer},
 ADDRESS = {Berlin}
}

@ARTICLE{barndorff-nielsen-hansen-lunde-shephard:08a,
 AUTHOR = {O. E. Barndorff-Nielsen and P. R. Hansen and A. Lunde and N. Shephard},
 YEAR = {2008},
 TITLE = {{Designing realized kernels to measure the ex post variation of equity prices in the presence of noise}},
 JOURNAL = {Econometrica},
 VOLUME = {76},
 NUMBER = {6},
 PAGES = {1481--1536}
}

@ARTICLE{barndorff-nielsen-hansen-lunde-shephard:09a,
 AUTHOR = {O. E. Barndorff-Nielsen and P. R. Hansen and A. Lunde and N. Shephard},
 YEAR = {2009},
 TITLE = {{Realized kernels in practice: trades and quotes}},
 JOURNAL = {Econometrics Journal},
 VOLUME = {12},
 NUMBER = {3},
 PAGES = {1--32}
}

@ARTICLE{barndorff-nielsen-shephard:02a,
 AUTHOR = {O. E. Barndorff-Nielsen and N. Shephard},
 YEAR = {2002},
 TITLE = {{Econometric analysis of realized volatility and its use in estimating stochastic volatility models}},
 JOURNAL = {Journal of the Royal Statistical Society: Series B},
 VOLUME = {64},
 NUMBER = {2},
 PAGES = {253--280}
}

@ARTICLE{barndorff-nielsen-shephard:04b,
 AUTHOR = {O. E. Barndorff-Nielsen and N. Shephard},
 YEAR = {2004},
 TITLE = {{Power and bipower variation with stochastic volatility and jumps}},
 JOURNAL = {Journal of Financial Econometrics},
 VOLUME = {2},
 NUMBER = {1},
 PAGES = {1--48}
}

@ARTICLE{barndorff-nielsen-shephard:06a,
 AUTHOR = {O. E. Barndorff-Nielsen and N. Shephard},
 YEAR = {2006},
 TITLE = {{Econometrics of testing for jumps in financial economics using bipower variation}},
 JOURNAL = {Journal of Financial Econometrics},
 VOLUME = {4},
 NUMBER = {1},
 PAGES = {1--30}
}

@INCOLLECTION{barndorff-nielsen-shephard:07a,
 AUTHOR = {O. E. Barndorff-Nielsen and N. Shephard},
 YEAR = {2007},
 TITLE = {{Variation, jumps, market frictions and high frequency data in financial econometrics}},
 BOOKTITLE = {Advances in Economics and Econometrics: Theory and Applications, Ninth World Congress, Volume III},
 EDITOR = {R. Blundell and P. Torsten and W. K. Newey},
 PAGES = {328--372},
 PUBLISHER = {Cambridge University Press},
 ADDRESS = {Cambridge}
}

@ARTICLE{barndorff-nielsen-shephard-winkel:06a,
 AUTHOR = {O. E. Barndorff-Nielsen and N. Shephard and M. Winkel},
 YEAR = {2006},
 TITLE = {{Limit theorems for multipower variation in the presence of jumps}},
 JOURNAL = {Stochastic Processes and their Applications},
 VOLUME = {116},
 NUMBER = {5},
 PAGES = {796--806}
}

@ARTICLE{bates:96a,
 AUTHOR = {D. S. Bates},
 YEAR = {1996},
 TITLE = {{Jumps and stochastic volatility: Exchange rate processes implicit in Deutsche Mark options}},
 JOURNAL = {Review of Financial Studies},
 VOLUME = {9},
 NUMBER = {1},
 PAGES = {69--107}
}

@ARTICLE{bates:12a,
 AUTHOR = {D. S. Bates},
 YEAR = {2012},
 TITLE = {{U.S. stock market crash risk, 1926 -- 2010}},
 JOURNAL = {Journal of Financial Economics},
 VOLUME = {105},
 NUMBER = {2},
 PAGES = {229--259}
}

@ARTICLE{bollerslev-law-tauchen:08a,
 AUTHOR = {T. Bollerslev and T. H. Law and G. Tauchen},
 YEAR = {2008},
 TITLE = {{Risk, jumps, and diversification}},
 JOURNAL = {Journal of Econometrics},
 VOLUME = {144},
 NUMBER = {1},
 PAGES = {234--256}
}

@ARTICLE{bollerslev-todorov:11a,
 AUTHOR = {T. Bollerslev and V. Todorov},
 YEAR = {2011},
 TITLE = {{Tails, fears and risk premia}},
 JOURNAL = {Journal of Finance},
 VOLUME = {66},
 NUMBER = {6},
 PAGES = {2165--2211}
}

@ARTICLE{brandt-diebold:06a,
 AUTHOR = {M. W. Brandt and F. X. Diebold},
 YEAR = {2006},
 TITLE = {{A no-arbitrage approach to range-based estimation of return covariances and correlations}},
 JOURNAL = {Journal of Business},
 VOLUME = {79},
 NUMBER = {1},
 PAGES = {61--73}
}

@ARTICLE{chernov-gallant-ghysels-tauchen:03a,
 AUTHOR = {M. Chernov and A. R. Gallant and E. Ghysels and G. Tauchen},
 YEAR = {2003},
 TITLE = {{Alternative models for stock price dynamics}},
 JOURNAL = {Journal of Econometrics},
 VOLUME = {116},
 NUMBER = {1--2},
 PAGES = {225--257}
}

@ARTICLE{christensen-oomen-podolskij:10a,
 AUTHOR = {K. Christensen and R. C. A. Oomen and M. Podolskij},
 YEAR = {2010},
 TITLE = {{Realised quantile-based estimation of the integrated variance}},
 JOURNAL = {Journal of Econometrics},
 VOLUME = {159},
 NUMBER = {1},
 PAGES = {74--98}
}

@ARTICLE{christensen-oomen-podolskij:14a,
 AUTHOR = {K. Christensen and R. C. A. Oomen and M. Podolskij},
 YEAR = {2014},
 TITLE = {{Fact or friction: Jumps at ultra high frequency}},
 JOURNAL = {Journal of Financial Economics},
 VOLUME = {114},
 NUMBER = {3},
 PAGES = {576--599}
}

@ARTICLE{christensen-podolskij:07a,
 AUTHOR = {K. Christensen and M. Podolskij},
 YEAR = {2007},
 TITLE = {{Realized range-based estimation of integrated variance}},
 JOURNAL = {Journal of Econometrics},
 VOLUME = {141},
 NUMBER = {2},
 PAGES = {323--349}
}

@ARTICLE{christensen-podolskij-vetter:09a,
 AUTHOR = {K. Christensen and M. Podolskij and M. Vetter},
 YEAR = {2009},
 TITLE = {{Bias-correcting the realized range-based variance in the presence of market microstructure noise}},
 JOURNAL = {Finance and Stochastics},
 VOLUME = {13},
 NUMBER = {2},
 PAGES = {239--268}
}

@ARTICLE{christie:82a,
 AUTHOR = {A. A. Christie},
 YEAR = {1982},
 TITLE = {{The stochastic behavior of common stock variances: Value, leverage and interest rate effects}},
 JOURNAL = {Journal of Financial Economics},
 VOLUME = {10},
 NUMBER = {4},
 PAGES = {407--432}
}

@ARTICLE{corsi-pirino-reno:10a,
 AUTHOR = {F. Corsi and D. E. Pirino and R. Ren\`{o}},
 YEAR = {2010},
 TITLE = {{Threshold bipower variation and the impact of jumps on volatility forecasting}},
 JOURNAL = {Journal of Econometrics},
 VOLUME = {159},
 NUMBER = {2},
 PAGES = {276--288}
}

@ARTICLE{corsi-reno:12a,
 AUTHOR = {F. Corsi and R. Ren\`{o}},
 YEAR = {2012},
 TITLE = {{Discrete-time volatility forecasting with persistent leverage effect and the link with continuous-time volatility modeling}},
 JOURNAL = {Journal of Business and Economic Statistics},
 VOLUME = {30},
 NUMBER = {3},
 PAGES = {368--380}
}

@ARTICLE{delbaen-schachermayer:94a,
 AUTHOR = {F. Delbaen and W. Schachermayer},
 YEAR = {1994},
 TITLE = {{A general version of the fundamental theorem of asset pricing}},
 JOURNAL = {Mathematische Annalen},
 VOLUME = {300},
 NUMBER = {1},
 PAGES = {463--520}
}

@TECHREPORT{dobrev:07a,
 AUTHOR = {D. Dobrev},
 TITLE = {{Capturing volatility from large price moves: Generalized range theory and applications}},
 INSTITUTION = {Kellogg School of Management, Northwestern University},
 YEAR = {2007},
 TYPE = {Working paper}
}

@TECHREPORT{dobrev-szerszen:10a,
 AUTHOR = {D. Dobrev and P. Szerszen},
 TITLE = {{The information content of high-frequency data for estimating equity return models and forecasting risk}},
 INSTITUTION = {Board of Governers of the Federal Reserve System},
 YEAR = {2010},
 TYPE = {Working paper}
}

@ARTICLE{eraker-johannes-polson:03a,
 AUTHOR = {B. Eraker and M. Johannes and N. Polson},
 YEAR = {2003},
 TITLE = {{The impact of jumps in volatility and returns}},
 JOURNAL = {Journal of Finance},
 VOLUME = {58},
 NUMBER = {3},
 PAGES = {1269--1300}
}

@ARTICLE{fan-wang:07a,
 AUTHOR = {J. Fan and Y. Wang},
 YEAR = {2007},
 TITLE = {{Multi-scale jump and volatility analysis for high-frequency financial data}},
 JOURNAL = {Journal of the American Statistical Association},
 VOLUME = {102},
 NUMBER = {480},
 PAGES = {1349--1362}
}

@ARTICLE{gallant-hsu-tauchen:99a,
 AUTHOR = {A. R. Gallant and C.-T. Hsu and G. E. Tauchen},
 YEAR = {1999},
 TITLE = {{Using daily range data to calibrate volatility diffusions and extract the forward integrated variance}},
 JOURNAL = {Review of Economics and Statistics},
 VOLUME = {81},
 NUMBER = {4},
 PAGES = {617--631}
}

@ARTICLE{goncalves-meddahi:09a,
 AUTHOR = {S. Gon\c{c}alves and N. Meddahi},
 YEAR = {2009},
 TITLE = {{Bootstrapping realized volatility}},
 JOURNAL = {Econometrica},
 VOLUME = {77},
 NUMBER = {1},
 PAGES = {283--306}
}

@ARTICLE{goncalves-meddahi:11a,
 AUTHOR = {S. Gon\c{c}alves and N. Meddahi},
 YEAR = {2011},
 TITLE = {{Box-Cox transforms for realized volatility}},
 JOURNAL = {Journal of Econometrics},
 VOLUME = {160},
 NUMBER = {1},
 PAGES = {129--144}
}

@ARTICLE{hansen-lunde:06a,
 AUTHOR = {P. R. Hansen and A. Lunde},
 YEAR = {2006},
 TITLE = {{Consistent ranking of volatility models}},
 JOURNAL = {Journal of Econometrics},
 VOLUME = {131},
 NUMBER = {1--2},
 PAGES = {97--121}
}

@ARTICLE{heston:93a,
 AUTHOR = {S. L. Heston},
 YEAR = {1993},
 TITLE = {{A closed-form solution for options with stochastic volatility with applications to bond and currency options}},
 JOURNAL = {Review of Financial Studies},
 VOLUME = {6},
 NUMBER = {2},
 PAGES = {327--343}
}

@ARTICLE{huang-tauchen:05a,
 AUTHOR = {X. Huang and G. Tauchen},
 YEAR = {2005},
 TITLE = {{The relative contribution of jumps to total price variance}},
 JOURNAL = {Journal of Financial Econometrics},
 VOLUME = {3},
 NUMBER = {4},
 PAGES = {456--499}
}

@ARTICLE{hull-white:87a,
 AUTHOR = {J. Hull and A. White},
 YEAR = {1987},
 TITLE = {{The pricing of options on assets with stochastic volatilities}},
 JOURNAL = {Journal of Finance},
 VOLUME = {42},
 NUMBER = {2},
 PAGES = {281--300}
}

@ARTICLE{jacod-li-mykland-podolskij-vetter:09a,
 AUTHOR = {J. Jacod and Y. Li and P. A. Mykland and M. Podolskij and M. Vetter},
 YEAR = {2009},
 TITLE = {{Microstructure noise in the continuous case: The pre-averaging approach}},
 JOURNAL = {Stochastic Processes and their Applications},
 VOLUME = {119},
 NUMBER = {7},
 PAGES = {2249--2276}
}

@BOOK{jacod-shiryaev:03a,
 AUTHOR = {J. Jacod and A. N. Shiryaev},
 YEAR = {2003},
 TITLE = {{Limit Theorems for Stochastic Processes}},
 EDITION = {2nd},
 PUBLISHER = {Springer},
 ADDRESS = {Berlin}
}

@ARTICLE{jiang-oomen:08a,
 AUTHOR = {G. J. Jiang and R. C. A. Oomen},
 YEAR = {2008},
 TITLE = {{Testing for jumps when asset prices are observed with noise -- A "swap variance" approach}},
 JOURNAL = {Journal of Econometrics},
 VOLUME = {144},
 NUMBER = {2},
 PAGES = {352--370}
}

@BOOK{karatzas-shreve:91a,
 AUTHOR = {I. Karatzas and S. E. Shreve},
 YEAR = {1991},
 TITLE = {{Brownian Motion and Stochastic Calculus}},
 EDITION = {2nd},
 PUBLISHER = {Springer},
 ADDRESS = {Berlin}
}

@ARTICLE{large:10a,
 AUTHOR = {J. Large},
 YEAR = {2010},
 TITLE = {{Estimating quadratic variation when quoted prices change by a constant increment}},
 JOURNAL = {Journal of Econometrics},
 VOLUME = {160},
 NUMBER = {1},
 PAGES = {2--11}
}

@ARTICLE{lee-mykland:08a,
 AUTHOR = {S. S. Lee and P. A. Mykland},
 YEAR = {2008},
 TITLE = {{Jumps in financial markets: A new nonparametric test and jump dynamics}},
 JOURNAL = {Review of Financial Studies},
 VOLUME = {21},
 NUMBER = {6},
 PAGES = {2535--2563}
}

@ARTICLE{lee-mykland:12a,
 AUTHOR = {S. S. Lee and P. A. Mykland},
 YEAR = {2012},
 TITLE = {{Jumps in equilibrium prices and market microstructure noise}},
 JOURNAL = {Journal of Econometrics},
 VOLUME = {168},
 NUMBER = {2},
 PAGES = {396--406}
}

@TECHREPORT{li:11a,
 AUTHOR = {J. Li},
 TITLE = {{Testing for jumps: A delta-hedging perspective}},
 INSTITUTION = {Princeton University},
 YEAR = {2011},
 TYPE = {Working paper}
}

@ARTICLE{mancini:04a,
 AUTHOR = {C. Mancini},
 YEAR = {2004},
 TITLE = {{Estimation of the characteristics of jump of a general Poisson-diffusion model}},
 JOURNAL = {Scandinavian Actuarial Journal},
 VOLUME = {2004},
 NUMBER = {1},
 PAGES = {42--52}
}

@ARTICLE{mancini:09a,
 AUTHOR = {C. Mancini},
 YEAR = {2009},
 TITLE = {{Non-parametric threshold estimation for models with stochastic diffusion coefficient and jumps}},
 JOURNAL = {Scandinavian Journal of Statistics},
 VOLUME = {36},
 NUMBER = {2},
 PAGES = {270--296}
}

@ARTICLE{martens-dijk:07a,
 AUTHOR = {M. Martens and D. van Dijk},
 YEAR = {2007},
 TITLE = {{Measuring volatility with the realized range}},
 JOURNAL = {Journal of Econometrics},
 VOLUME = {138},
 NUMBER = {1},
 PAGES = {181--207}
}

@ARTICLE{mykland:10a,
 AUTHOR = {P. A. Mykland},
 YEAR = {2010},
 TITLE = {{A Gaussian calculus for inference from high frequency data}},
 JOURNAL = {Annals of Finance},
 VOLUME = {8},
 NUMBER = {2},
 PAGES = {235--258}
}

@MISC{mykland-shephard-sheppard:12a,
 AUTHOR = {P. A. Mykland and N. Shephard and K. Sheppard},
 YEAR = {2012},
 TITLE = {{Efficient and feasible inference for the components of financial variation using blocked multipower variation}},
 NOTE = {Working paper}
}

@ARTICLE{oomen:06a,
 AUTHOR = {R. C. A. Oomen},
 YEAR = {2006},
 TITLE = {{Comment on 2005 JBES invited address ``Realized variance and market microstructure noise'' by Peter R. Hansen and Asger Lunde}},
 JOURNAL = {Journal of Business and Economic Statistics},
 VOLUME = {24},
 NUMBER = {2},
 PAGES = {195--202}
}

@ARTICLE{parkinson:80a,
 AUTHOR = {M. Parkinson},
 YEAR = {1980},
 TITLE = {{The extreme value method for estimating the variance of the rate of return}},
 JOURNAL = {Journal of Business},
 VOLUME = {53},
 NUMBER = {1},
 PAGES = {61--65}
}

@ARTICLE{patton-sheppard:09a,
 AUTHOR = {A. J. Patton and K. Sheppard},
 YEAR = {2009},
 TITLE = {{Optimal combinations of realised volatility estimators}},
 JOURNAL = {International Journal of Forecasting},
 VOLUME = {25},
 NUMBER = {2},
 PAGES = {218--238}
}

@ARTICLE{podolskij-vetter:09a,
 AUTHOR = {M. Podolskij and M. Vetter},
 YEAR = {2009},
 TITLE = {{Bipower-type estimation in a noisy diffusion setting}},
 JOURNAL = {Stochastic Processes and their Applications},
 VOLUME = {119},
 NUMBER = {9},
 PAGES = {2803--2831}
}

@ARTICLE{podolskij-vetter:09b,
 AUTHOR = {M. Podolskij and M. Vetter},
 YEAR = {2009},
 TITLE = {{Estimation of volatility functionals in the simultaneous presence of microstructure noise and jumps}},
 JOURNAL = {Bernoulli},
 VOLUME = {15},
 NUMBER = {3},
 PAGES = {634--658}
}

@BOOK{protter:04a,
 AUTHOR = {P. E. Protter},
 YEAR = {2004},
 TITLE = {{Stochastic Integration and Differential Equations}},
 EDITION = {1st},
 PUBLISHER = {Springer},
 ADDRESS = {Berlin}
}

@ARTICLE{renyi:63a,
 AUTHOR = {A. R\'{e}nyi},
 YEAR = {1963},
 TITLE = {{On stable sequences of events}},
 JOURNAL = {Sankhya: The Indian Journal of Statistics; Series A},
 VOLUME = {25},
 NUMBER = {3},
 PAGES = {293--302}
}

@MISC{rossi-spazzini:09a,
 AUTHOR = {E. Rossi and F. Spazzini},
 YEAR = {2009},
 TITLE = {{Finite sample results of range-based integrated volatility estimation}},
 NOTE = {Working paper}
}

@ARTICLE{tauchen-zhou:11a,
 AUTHOR = {G. E. Tauchen and H. Zhou},
 YEAR = {2011},
 TITLE = {{Realized jumps on financial markets and predicting credit spreads}},
 JOURNAL = {Journal of Econometrics},
 VOLUME = {160},
 NUMBER = {1},
 PAGES = {102--118}
}

@ARTICLE{todorov:09a,
 AUTHOR = {V. Todorov},
 YEAR = {2009},
 TITLE = {{Estimation of continuous-time stochastic volatility models with jumps using high-frequency data}},
 JOURNAL = {Journal of Econometrics},
 VOLUME = {148},
 NUMBER = {2},
 PAGES = {131--148}
}

@ARTICLE{todorov:10a,
 AUTHOR = {V. Todorov},
 YEAR = {2010},
 TITLE = {{Variance risk-premium dynamics: The role of jumps}},
 JOURNAL = {Review of Financial Studies},
 VOLUME = {23},
 NUMBER = {1},
 PAGES = {345--383}
}

@ARTICLE{todorov-tauchen:10a,
 AUTHOR = {V. Todorov and G. Tauchen},
 YEAR = {2010},
 TITLE = {{Activity signature functions for high-frequency data analysis}},
 JOURNAL = {Journal of Econometrics},
 VOLUME = {154},
 NUMBER = {2},
 PAGES = {125--138}
}

@ARTICLE{wasserfallen-zimmermann:85a,
 AUTHOR = {W. Wasserfallen and H. Zimmermann},
 YEAR = {1985},
 TITLE = {{The behavior of intraday exchange rates}},
 JOURNAL = {Journal of Banking and Finance},
 VOLUME = {9},
 NUMBER = {1},
 PAGES = {55--72}
}

@ARTICLE{woerner:05a,
 AUTHOR = {J. H. C. Woerner},
 YEAR = {2005},
 TITLE = {{Estimation of integrated volatility in stochastic volatility models}},
 JOURNAL = {Applied Stochastic Models in Business and Industry},
 VOLUME = {21},
 NUMBER = {1},
 PAGES = {27--44}
}

@INPROCEEDINGS{woerner:06a,
 AUTHOR = {J. H. C. Woerner},
 TITLE = {{Power and multipower variation: Inference for high-frequency data}},
 BOOKTITLE = {Proceedings of the International Conference on Stochastic Finance 2004},
 YEAR = {2006},
 EDITOR = {A. N. Shiryaev and M. do R. Grossihno and P. Oliviera and M. Esquivel},
 PAGES = {343--364},
 PUBLISHER = {Springer},
 ADDRESS = {Berlin},
}

@ARTICLE{zhang:06a,
 AUTHOR = {L. Zhang},
 YEAR = {2006},
 TITLE = {{Efficient estimation of stochastic volatility using noisy observations: A multi-scale approach}},
 JOURNAL = {Bernoulli},
 VOLUME = {12},
 NUMBER = {6},
 PAGES = {1019--1043}
}

@ARTICLE{zhang-mykland-ait-sahalia:05a,
 AUTHOR = {L. Zhang and P. A. Mykland and Y. A\"{i}t-Sahalia},
 YEAR = {2005},
 TITLE = {{A tale of two time scales: determining integrated volatility with noisy high-frequency data}},
 JOURNAL = {Journal of the American Statistical Association},
 VOLUME = {100},
 NUMBER = {472},
 PAGES = {1394--1411}
}

@ARTICLE{zhou:96a,
 AUTHOR = {B. Zhou},
 YEAR = {1996},
 TITLE = {{High-frequency data and volatility in foreign-exchange rates}},
 JOURNAL = {Journal of Business and Economic Statistics},
 VOLUME = {14},
 NUMBER = {1},
 PAGES = {45--52}
}

\end{document}